\documentclass[
  onecolumn,
  superscriptaddress,
  a4paper,
  11pt,
  accepted=2024-12-27]{quantumarticle}

\pdfoutput=1

\usepackage{authblk}

\usepackage[margin=1in]{geometry}

\usepackage{enumerate}
\usepackage{enumitem}

\usepackage{microtype}

\usepackage{tikz}

\usepackage{ytableau}

\usepackage[normalem]{ulem}
\usepackage{float}
\usepackage[font=small]{caption}

\usepackage[colorlinks = true]{hyperref}
\definecolor{darkred}  {rgb}{0.5,0,0}
\definecolor{darkblue} {rgb}{0,0,0.5}
\definecolor{darkgreen}{rgb}{0,0.5,0}
\hypersetup{
  urlcolor   = blue,         
  linkcolor  = darkblue,     
  citecolor  = darkgreen,    
  filecolor  = darkred       
}

\usepackage{amsmath,amssymb,amsfonts,amsthm,amstext}

\usepackage{etoolbox}

\usepackage{mathtools}
\mathtoolsset{centercolon}
\makeatletter
\protected\def\tikz@nonactivecolon{\ifmmode\mathrel{\mathop\ordinarycolon}\else:\fi}
\makeatother

\usepackage{tcolorbox}

\usepackage{algorithm}
\usepackage[beginComment=]{algpseudocodex}

\usepackage[nameinlink]{cleveref}
\crefname{lemma}{Lemma}{Lemmas}
\crefname{proposition}{Proposition}{Propositions}
\crefname{definition}{Definition}{Definitions}
\crefname{theorem}{Theorem}{Theorems}
\crefname{corollary}{Corollary}{Corollaries}
\crefname{claim}{Claim}{Claims}
\crefname{claimalt}{Claim}{Claims}
\crefname{section}{Section}{Sections}
\crefname{appendix}{Appendix}{Appendices}
\crefname{figure}{Fig.}{Figs.}
\crefname{table}{Table}{Tables}
\crefname{algorithm}{Algorithm}{Algorithms}
\crefname{enumi}{part}{parts}

\usepackage[retainorgcmds]{IEEEtrantools}

\usepackage[titles]{tocloft}

\usepackage{placeins}

\newcommand{\ket}[1]{|{#1}\rangle}
\newcommand{\bra}[1]{\langle{#1}|}

\newcommand{\ketbra}[2]{|{#1}\rangle\langle{#2}|}
\newcommand{\proj}[1]{\ketbra{#1}{#1}}

\newcommand{\x}{\otimes}
\newcommand{\xp}[1]{^{\otimes #1}}

\DeclarePairedDelimiter{\set}{\lbrace}{\rbrace}

\DeclarePairedDelimiter{\of}{\lparen}{\rparen}

\DeclarePairedDelimiter{\ceiling}{\lceil}{\rceil}

\let\originalleft\left
\let\originalright\right
\renewcommand{\left}{\mathopen{}\mathclose\bgroup\originalleft}
\renewcommand{\right}{\aftergroup\egroup\originalright}

\usepackage{trimclip}

\makeatletter 
\newcommand*\unevendelim[3]{{\mathpalette\unevendelim@{{#1}{#2}{#3}}}}
\def\unevendelim@#1#2{\unevendelim@@{#1}#2}%
\def\unevendelim@@#1#2#3#4{%
    \sbox0{$\m@th#1#4$}%
    \sbox6{$\m@th#1\{\}$}%
    \unevendelim@@@{#1}{\left#2}{\right.}%
    \copy0
    \unevendelim@@@{#1}{\left.}{\right#3}%
}
\def\unevendelim@@@#1#2#3{%
  \sbox2{$\m@th#1#2\rule{0pt}{\ht0}#3$}%
  \sbox4{$\m@th#1#2\rule[-\dp0]{0pt}{\dp0}#3$}%
  \ifdim\ht2>\ht4
    \ooalign{\clipbox{0pt {\dimexpr\dp2-\dp6\relax} 0pt 0pt}{\copy2}\cr
             \raisebox{-\dp6}{\clipbox{0pt 0pt 0pt {\dimexpr2\ht2-\ht4+\dp6}}{\copy2}}\cr}
  \else
    \ooalign{\raisebox{\ht6}{\clipbox{0pt {\dimexpr2\dp4-\dp2+\ht6} 0pt 0pt}{\copy4}}\cr
             \clipbox{0pt 0pt 0pt {\dimexpr\ht4-\ht6}}{\copy4}\cr}
  \fi
}
\makeatother  

\DeclareMathOperator{\Tr}{Tr}

\DeclareMathOperator{\poly}{poly}

\newcommand{\SWAP}{\textsc{Swap}}
\newcommand{\PURIFY}{\textsc{Purify}}

\newcommand{\smx}[1]{\bigl(\begin{smallmatrix}#1\end{smallmatrix}\bigr)}

\newcommand{\C}{\mathbb{C}}

\newcommand{\U}[1]{\mathrm{U}(#1)}

\renewcommand{\S}[1]{\mathrm{S}_{#1}}

\newcommand{\SC}{\mathrm{C}} 

\newcommand{\e}{\delta}

\newcommand{\D}{\Delta}

\usepackage{color}

\usepackage[
  bordercolor = white,
  backgroundcolor = gray!30,
  linecolor = black,
  colorinlistoftodos]{todonotes}

\newtheorem{theorem}{Theorem}
\newtheorem{lemma}[theorem]{Lemma}
\newtheorem{proposition}[theorem]{Proposition}

\newtheorem{corollary}[theorem]{Corollary}
\newtheorem{claim}{Claim}

\newenvironment{claimproof}[1][\proofname]{\proof[#1]}{\endproof}

\newtheorem{claimalt}{Claim}[claim]

\newenvironment{claimd}[1]{
  \renewcommand\theclaimalt{#1}
  \claimalt
}{\endclaimalt}

\newcommand{\w}{0.5cm}

\newcommand{\bx}[3]{
  \draw #3 (#1*\w-\w/2,-#2*\w-\w/2) rectangle (#1*\w+\w/2,-#2*\w+\w/2);
}

\newcommand{\mbx}[3]{
  \bx{#1}{#2}{\ifodd #3 [fill = lightgray] \fi}
  \ifnum #3 > 3 \fill (#1*\w,-#2*\w) circle [radius = 0.15*\w]; \fi
  \ifnum #3 = 3 \draw (#1*\w,-#2*\w) node {$x$}; \fi
}

\makeatletter
\newcommand{\nbx}[3]{\onebox{#1}{#2}{\expandafter\nodeaux#3!\@nil}{\expandafter\redaux#3!\@nil}}
\newcommand{\onebox}[4]{
  \bx{#1}{#2}{#4}
  \draw (#1*\w,-#2*\w) node {#3};
}
\def\nodeaux#1!#2\@nil{%
  \if\relax\detokenize{#2}\relax
    \expandafter\@firstoftwo
  \else
    \expandafter\@secondoftwo
  \fi
  {#1}%
  {#1}%
}
\def\redaux#1!#2\@nil{%
  \if\relax\detokenize{#2}\relax
    \expandafter\@firstoftwo
  \else
    \expandafter\@secondoftwo
  \fi
  {}%
  {[fill=lightgray]}%
}
\makeatother

\usepackage{parskip}
\begin{document}

\title{Streaming quantum state purification}

\author[1,2]{Andrew M.\ Childs}
\author[1,2,3,6]{Honghao Fu}
\author[1,4]{Debbie Leung}
\author[4]{Zhi Li}
\author[5]{Maris Ozols}
\author[1]{Vedang Vyas}


\affil[1]{{\footnotesize Institute for Quantum Computing, University of Waterloo,
    Waterloo, ON, N2L 3G1, Canada}}
\affil[2]{Department of Computer Science, Institute for Advanced Computer Studies, and Joint Center for Quantum Information and Computer Science, University of Maryland, College Park, MD 20742, USA}
\affil[3]{Massachusetts Institute of Technology, 77 Massachusetts Ave., Cambridge, MA 02139, USA}
\affil[4]{Perimeter Institute for Theoretical Physics, 31 Caroline St.\ N., Waterloo, ON N2L 2Y5, Canada}
\affil[5]{QuSoft and University of Amsterdam, Science Park 123, 1098 XG, Amsterdam, the Netherlands}
\affil[6]{Concordia Institute of Information Systems Engineering, Concordia University, 1455 Blvd.\ De Maisonneuve Ouest, Montreal, QC H3G 1M8, Canada}

\date{}

\maketitle

\vspace*{-3ex}

\begin{abstract}
Quantum state purification is the task of recovering a nearly pure copy of an unknown pure quantum state using multiple noisy copies of the state.
This basic task has applications to quantum communication over noisy channels and quantum computation with imperfect devices, but has primarily 
been studied previously for the case of qubits.
We derive an efficient purification procedure based on the swap test for qudits of any dimension, starting with any initial error parameter.  Treating the initial error parameter and the dimension as constants, we show that our procedure has sample complexity asymptotically optimal in the final error parameter.
Our protocol has a simple recursive structure that can be applied when the states are provided one at a time in a streaming fashion, requiring only a small quantum memory to implement.
\end{abstract}

\setlength{\cftbeforesecskip}{4pt}
\renewcommand\cftsecafterpnum{\vskip2pt}
\setcounter{tocdepth}{2}

\vspace*{-4ex}

\tableofcontents

\section{Introduction}

Quantum states are notoriously susceptible to the effects of decoherence. Thus, a basic challenge in quantum information processing is to find ways of protecting quantum systems from noise, or of removing noise that has already occurred. In this paper we study the latter approach, aiming to (partially) reverse the effect of decoherence and produce less noisy states out of mixed ones, a process that we call \emph{purification}.

Since noise makes quantum states less distinguishable, it is impossible to purify a single copy of a noisy state. However, given multiple copies of a noisy state, we can try to reconstruct a single copy that is closer to the original pure state. In general, we expect the quality of the reconstructed state to be higher if we are able to use more copies of the noisy state. We would like to understand the number of samples that suffice to produce a high-fidelity copy of the original state, and to give efficient procedures for carrying out such purification.

More precisely, we focus on depolarizing noise and consider the following formulation of the purification problem.
Let $\ket{\psi} \in \C^d$ be an unknown pure $d$-dimensional state (or \emph{qudit}).
In the \emph{qudit purification problem}, we are given multiple noisy copies of $\ket{\psi}$ that are all of the form
\begin{equation}
  \rho(\e)
  := (1 - \e) \proj{\psi}
  + \e \, \frac{I}{d},
  \label{eq:re}
\end{equation}
where the \emph{error parameter}
$\e \in (0,1)$ represents the probability that the state is depolarized.
If we need to compare error parameters across different dimensions, we write $\e^{(d)}$ to specify the underlying dimension.

Let $\mathcal{P}$ denote a purification procedure, which is a
quantum operation that maps
$N$ copies of $\rho(\delta)$ to a single qudit. (Note that the implementation of $\mathcal{P}$ can make use of additional quantum workspace.) We refer to $N$ as the \emph{sample complexity} of $\mathcal{P}$. The goal is to produce a single qudit that closely approximates the desired (but unknown) target state $\ket{\psi}$.
In particular, we aim to produce a state whose fidelity with the ideal pure state $\ket{\psi}$ is within some specified tolerance, using the smallest possible sample complexity $N$. Our streaming protocol produces a state of the form $\rho(\epsilon)$ for some small $\epsilon>0$, which has fidelity $1-(1-\frac{1}{d})\epsilon$ with $\ket{\psi}$, so it is convenient to focus on the output error parameter $\epsilon$.

The purification problem was first studied over two decades ago.
Inspired by studies of entanglement purification \cite{B96}, Cirac, Ekert, and Macchiavello studied the problem of purifying a depolarized qubit and found the optimal purification procedure \cite{CEM99}. Later, Keyl and Werner studied qubit purification under different criteria, including the possibility of producing multiple copies of the output state and measuring the fidelity either by comparing a single output state or selecting all the output states \cite{KW01}.  Subsequent work studied probabilistic quantum state purification using semidefinite programming \cite{F04}. 
However, until recently, little was known about how to purify higher-dimensional quantum states.

Our main contribution is the first concrete and efficient qudit purification procedure, which achieves the following.
\begin{theorem}[Informal]\label{thm:informal}
	For any fixed input error parameter $\e \in (0,1)$, there is a protocol that produces a state with output error parameter at most $\epsilon \in (0,1)$ using $O(1/\epsilon)$ samples of the $d$-dimensional input state and $O(\frac{1}{\epsilon} \log d)$ elementary gates.
\end{theorem}
\noindent
We present this protocol in \cref{sec:swap} and analyze its performance in \cref{sec:analysis}, establishing a formal version of the above result (\cref{thm:swap} and \cref{cor:gate_comp}).

In \cref{sec:apps}, we describe several applications of our results.  
In \cref{sec:simon}, we use our procedure to solve a version of Simon's problem with a faulty oracle. In particular, we show that if the oracle depolarizes its output by a constant amount, then the problem can still be solved with only quadratic overhead, preserving the exponential quantum speedup.
In \cref{sec:mt-qst}, we discuss the relationship between mixedness testing
(introduced and studied in \cite{MdW13,DW15,Wright-thesis-16}), quantum state purification, and quantum
state tomography. In some sense, these problems are increasingly difficult.  Our protocol implies a dimension-independent upper bound for
the sample complexity of mixedness testing for depolarized states,
below the worst-case lower bound shown in \cite{DW15}. 
On the other hand,
quantum state tomography can be used for quantum state purification,
even in a streaming fashion. We compare the performance of this approach and purification based on the
recursive swap test. The swap test is more efficient for 
large dimension and fixed input noise, whereas state tomography is more 
efficient for small dimension and very small signal.
In \cref{sec:q-majority}, we briefly discuss the relationship between quantum state purification and the quantum majority vote problem
for qubits, including how the former provides a higher-dimensional
generalization of the latter.

Rather than applying Schur--Weyl duality globally as in \cite{CEM99}, our streaming purification protocol recursively applies the swap test to pairs of quantum states.
This approach has two advantages. First, our procedure readily works for quantum states of any dimension.
Second, the swap test is much easier to implement and analyze than a general Schur basis measurement. In particular, the protocol can be implemented in a scenario where the states arrive in an online fashion, using a quantum memory of size only logarithmic in the total number of states used.

Since our protocol does not respect the full permutation symmetry of the problem, it is not precisely optimal. However, suppose we fix the dimension and the initial error parameter. Then, in the qubit case, our sample complexity matches that of the optimal protocol \cite{CEM99} up to factors involving these constants, suggesting that the procedure performs well. Furthermore, in \cref{sec:lowbdd} we prove a lower bound showing that our protocol has optimal scaling in the final error parameter up to factors involving the other constants.

The optimal purification protocol should respect the permutation symmetry of the input and hence can be characterized using Schur--Weyl duality. Some of the authors developed a framework for studying such optimal purification procedures, which led to a recently reported optimal protocol \cite{li2024optimal} (discussed further below).
We briefly describe this related work in \cref{sec:con} along with a summary of the results and a discussion of some open questions.

Since this manuscript was first made available as an arXiv preprint, our purification protocol has been applied to more general noisy states, with applications in the problem of principal eigenstate classical shadows \cite{GPS24}. Additional results connecting purification and shadow estimation can be found in \cite{zhou-liu-22,LLYZZ-2024}.  
Furthermore, alternative streaming purification protocols have been developed using semidefinite programming and block encoding techniques \cite{YCHCW24,yang2024quantum}. 

After our initial journal submission, \cite{li2024optimal} reported an optimal purification procedure for arbitrary qudit states using Schur--Weyl duality.  In particular, the optimal sample complexity is now known for arbitrary dimension.  A comparison with this optimal result shows that our method based on the recursive swap test has optimal expected sample complexity as a function of the dimension (in addition to the final error parameter), up to a constant factor, for initial error $\delta < 1/2$. Since the optimal protocol involves performing the Schur transform \cite{BCH06}  and applying a change of basis between different irreducible representations of the unitary group, it is unclear whether the optimal protocol can be implemented efficiently. 
Even if it can, it is unlikely to be practical. 
Moreover, the optimal protocol cannot be used as a gadget in the streaming model because its output is not a depolarized state, unlike the swap-test gadget. Hence, it requires a large quantum memory to store $N$ input states. In contrast, our streaming protocol only needs a quantum memory of size $O(\log(N))$ and very simple operations.
Therefore, our approach will likely be the method of choice in practice.

\section{Purification using the swap test}\label{sec:swap}

In this section, we present a recursive purification procedure based on the swap test.
We review the swap test in \cref{sec:swaptest} and define a gadget that uses it to project onto the symmetric subspace in \cref{sec:swapgadget}, and then we specialize the analysis to the case where the two input copies have the same purity. Next, we define our recursive purification procedure in \cref{sec:recursive}. Finally, we describe a stack-based implementation of this procedure in \cref{sec:stackbased}, providing a bound on the space complexity.

\subsection{The swap test} \label{sec:swaptest}

The \emph{swap test}~\cite{BCWD01} is given by the circuit shown in \cref{fig:swap}.
\newcommand{\round}{.. controls +(0.3,0) and +(-0.3,0) ..}
\begin{figure}[ht]
\centering
\begin{tikzpicture}[thick,
  box/.style = {fill = white, draw = black, inner sep = 0pt, text width = 0.55cm, text height = 0.55cm},
  arc/.style = {start angle = 90, end angle = 0, radius = 0.3cm},
  ell/.style = {x radius = 0.2cm, y radius = 0.1cm}]
  \def\w{0.4cm}
  \def\h{0.7cm}
  \def\d{0.4cm}
  \def\r{0.06cm}
  \path coordinate (0) at (-4*\w,2*\h);
  \path coordinate (1) at (-4*\w,1*\h);
  \path coordinate (2) at (-4*\w,0*\h);
  \path (0)+(-\d,0) node {$\ket{0}$};
  \path (1)+(-\d,0) node {$\rho$};
  \path (2)+(-\d,0) node {$\sigma$};
  \draw (0) -- ++(8*\w,0);
  \draw (1) -- ++(3*\w,0) \round ++(2*\w,-\h) -- ++(3*\w,0) arc[arc] circle[ell] ++(0.2cm,0) -- ++(-\r,-0.35cm) arc [start angle = 0, end angle = -180, x radius = 0.2cm-\r, y radius = 0.05cm] -- ++(-\r,0.35cm);
  \draw (2) -- ++(3*\w,0) \round ++(2*\w,+\h) -- ++(3*\w,0) -- ++(2*\w,0) ++(\d,0) node {\quad$\omega(a)$};
  \draw[double] (0)++(8*\w,0) -- ++(2*\w,0) ++(\d,0) node {$a$};
  \path (0)+(1.5*\w,0) node[box] {} node {$H$};
  \path (0)+(6.5*\w,0) node[box] {} node {$H$};
  \path (0)+(8.5*\w,0) node[box] {};
  \draw (0)++(8.5*\w+0.14cm,-0.1cm) arc [start angle = 0, end angle = 180, radius = 0.14cm];
  \draw (0)++(8.5*\w,-0.1cm) -- ++(0.1cm,0.2cm);
  \fill (0,2*\h) circle (0.08);
  \draw (0,2*\h) -- (0,0.5*\h);
\end{tikzpicture}
\caption{\label{fig:swap} Quantum circuit for the swap test on arbitrary input states $\rho$ and $\sigma$. The gate in the middle denotes the controlled-swap operation while $H = \frac{1}{\sqrt{2}} \smx{1&\phantom{-}1\\1&-1}$ is the Hadamard gate.}
\end{figure}
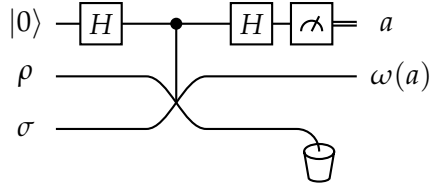
It takes two input registers 
of the same dimension, and requires one ancilla qubit initialized to~$\ket{0}$.  
The gate after the first Hadamard gate is the controlled-swap gate
$\proj{0} \x I + \proj{1} \x S$, where $I$ and $S$ denote respectively
the identity gate and the \emph{swap gate} ($\ket{i} \x \ket{j}
\mapsto \ket{j} \x \ket{i}$) acting on the last two registers.
Since the control qubit is prepared in the state $H \ket{0} = \frac{\ket{0}+\ket{1}}{\sqrt{2}}$ and then measured in the Hadamard basis $\frac{\ket{0}\pm\ket{1}}{\sqrt{2}}$, we can use the identity
\begin{equation}
\begin{tikzpicture}[thick, baseline = 0.3cm]
  \def\w{0.4cm}
  \def\h{0.7cm}
  \def\d{0.7cm}
  \path coordinate (0) at (-1*\w,2*\h);
  \path coordinate (1) at (-1*\w,1*\h);
  \path coordinate (2) at (-1*\w,0*\h);
  \path (0)+(    -\d,0) node {$\frac{\ket{0} + \ket{1}}{\sqrt{2}}$};
  \path (0)+(2*\w+\d,0) node {$\frac{\ket{0}\pm\ket{1}}{\sqrt{2}}$};
  \draw (0) -- ++(2*\w,0);
  \draw (1) \round ++(2*\w,-\h);
  \draw (2) \round ++(2*\w,+\h);
  \fill (0,2*\h) circle (0.08);
  \draw (0,2*\h) -- (0,0.5*\h);
  \path coordinate (R1) at (7*\w,1*\h);
  \path coordinate (R2) at (7*\w,0*\h);
  \path (R1)+(-2.3*\w,-0.5*\h) node {$=\;\;\displaystyle\frac{1}{2}$};
  \draw (R1) -- ++(2*\w,0);
  \draw (R2) -- ++(2*\w,0);
  \path (R1)+(3*\w,-0.5*\h) node {$\pm$};
  \draw (R1)++(4*\w,0) \round ++(2*\w,-\h);
  \draw (R2)++(4*\w,0) \round ++(2*\w,+\h);
  \path (R1)+( 6.6*\w,-0.5*\h) node {$\Bigg)$};
  \path (R1)+(-0.6*\w,-0.5*\h) node {$\Bigg($};
  \path (R1)+(10.5*\w,-0.5*\h) node {$=\;\;\displaystyle\frac{1}{2}(I \pm S)$};
\end{tikzpicture}
\end{equation}
to eliminate the control qubit when analyzing the output.  
Thus for joint input $\mu$ and measurement outcome $a \in \set{0,1}$ (see \cref{fig:swap}), the corresponding sub-normalized post-measurement state on the last two registers is
\begin{equation}
  \frac{I + (-1)^a S}{2} \; \mu \; \frac{I + (-1)^a S}{2},
  \label{eq:spoutboth}
\end{equation}
and the trace gives the probability of the outcome $a$.  
Note that $\frac{1}{2}(I+S)$ and $\frac{1}{2}(I-S)$ project onto the \emph{symmetric} and the \emph{anti-symmetric subspace}, respectively.  

The swap test is a standard procedure for comparing two unknown \emph{pure} quantum states
$\ket{\psi_{1,2}}$; 
in this case, $\mu = \proj{\psi_1} \otimes \proj{\psi_2}$.  
The probability to obtain the outcome $a=0$ decreases from $1$ to $1/2$ as the fidelity between the inputs decreases from $1$ to $0$ (i.e, as the inputs change from being identical to being orthogonal).  
If we instead apply the swap test to two copies of a mixed state $\rho$ (i.e.,
$\mu = \rho \x \rho$), the postmeasurement state
for the $a=0$ outcome has increased weights for matching eigenvectors on the two inputs, suggesting that the swap test can be used for purification.  
The swap test also potentially entangles the two output quantum registers, so our purification
algorithm keeps only one of the two quantum output registers.

We now make this intuition
concrete.  Starting from a product input $\mu = \rho \x \sigma$ in \cref{eq:spoutboth} and
discarding the last register yields the sub-normalized state
\begin{equation}
	\frac{1}{4} \of[\big]{\rho + \sigma + (-1)^a \rho \sigma + (-1)^a \sigma \rho},
  \label{eq:final_state}
\end{equation}
where the last two terms are obtained from the identity
$\Tr_2 [(\rho \otimes \sigma) S] = \rho \sigma$. 

The probability of obtaining outcome $a$ is given by the trace of \cref{eq:final_state},
\begin{equation}
  \frac{1}{2} \of[\big]{1 + (-1)^a \Tr(\rho\sigma)},
  \label{eq:trace}
\end{equation}
and the normalized output state $\omega(a)$ in \cref{fig:swap} is
\begin{equation}
  \omega(a)
  = \frac{1}{2} \cdot
    \frac{\rho + \sigma + (-1)^a (\rho \sigma + \sigma \rho)}
         {1 + (-1)^a \Tr(\rho\sigma)}.
  \label{eq:omega}
\end{equation}
Note that the probability of the $a=0$ outcome is at least $1/2$ for
a product input.  
Note also that the swap test can be implemented with $O(\log d)$ quantum gates since a swap gate between qudits can be decomposed into swaps of constituent qubits.

\subsection{The swap test gadget} \label{sec:swapgadget}

We are particularly interested in the output state $\omega(0)$ that corresponds to the measurement outcome $a = 0$.
This state is equal to the normalized projection of $\rho \x \sigma$ onto the symmetric subspace.
\Cref{alg:swap} introduces a \emph{$\SWAP$ gadget} that performs the swap test on fresh copies of the input states $\rho$ and $\sigma$ until it succeeds in obtaining the $a = 0$ outcome.

\begin{algorithm}[ht]
\caption{$\SWAP$ gadget.}
\label{alg:swap}
\emph{This procedure has access to a stream of inputs $\rho$ and $\sigma$, and uses as many copies of them as necessary.}
\begin{algorithmic}[1]
\Procedure{$\SWAP$}{$\rho,\sigma$}
\Repeat
  \State Ask for one fresh copy of $\rho$ and $\sigma$.
  \State Apply the swap test shown in \cref{fig:swap} and denote the measurement outcome by $a$.
\Until $a=0$
\State \Return the state $\omega(0)$ of the output register
\EndProcedure
\end{algorithmic}
\end{algorithm}

We denote the final output of the $\SWAP$ gadget by
\begin{equation}
  \SWAP(\rho,\sigma)
  := \frac{1}{2} \cdot
     \frac{\rho + \sigma + \rho \sigma + \sigma \rho}
          {1 + \Tr(\rho\sigma)},
  \label{eq:SWAP}
\end{equation}
which is just $\omega(0)$ from \cref{eq:omega}. The expected number of copies of each input state consumed by this gadget is given by the inverse of its success probability. By \cref{eq:trace}, this is $2/\of{1 + \Tr(\rho\sigma)}$.

We focus on applying the swap test gadget when $\rho = \sigma$.  
From \cref{eq:trace}, the outcome $a = 0$ in the swap test occurs with probability
\begin{equation}
  \frac{1}{2} \of[\big]{1 + \Tr (\rho^2)}
  \label{eq:prob}
\end{equation}
and, according to \cref{eq:SWAP}, the normalized output state is
\begin{equation}
  \SWAP(\rho,\rho) = \frac{\rho + \rho^2}{1+\Tr(\rho^2)}.
  \label{eq:rout}
\end{equation}
Note that $\SWAP(\rho,\rho)$ shares the same eigenvectors with $\rho$ and $\rho^2$.   
If the eigenvalues of $\rho$ are $\lambda_1 > \lambda_2 \geq \lambda_2 \geq \cdots \geq \lambda_d$, 
then the eigenvalues of $\SWAP(\rho,\rho)$ are 
\begin{equation}
\lambda_i' = \frac{\lambda_i + \lambda_i^2}{1+\sum_{j=1}^d \lambda_j^2} .  
\label{eq:evals}
\end{equation}
In particular, 
\begin{align}
\lambda_1' - \lambda_1 
& = \frac{\lambda_1 + \lambda_1^2}{1+\sum_{j=1}^d \lambda_j^2} - \lambda_1
= \frac{\lambda_1^2  \sum_{j=1}^d \lambda_j  - \lambda_1 \sum_{j=1}^d \lambda_j^2}{1+\sum_{j=1}^d \lambda_j^2} 
= \frac{ \lambda_1  \sum_{j=2}^d (\lambda_1 - \lambda_j) \lambda_j }{1+\sum_{j=1}^d \lambda_j^2},
\end{align}
which is positive unless $\rho$ is pure or maximally mixed.  
In particular, the following corollary holds for the noisy states we want to purify.  
\begin{corollary}\label{cor:better}
If $\e \in (0,1)$ then $\SWAP(\rho(\e),\rho(\e)) = \rho(\e')$ where $\e' < \e$.
\end{corollary}

This corollary suggests a simple recursive algorithm that applies the $\SWAP$ gadget on a pair of identical inputs produced in the previous level of the recursion. \Cref{cor:better} then guarantees that at each subsequent level of recursion, the states become purer.  Before we formally state and analyze this recursive procedure, let us first derive the success probability of the swap test on two copies of $\rho(\e)$, the expected number of copies of $\rho(\e)$ consumed by the $\SWAP$ gadget, and the error parameter $\e'$ of the output state $\rho(\e')$.
The success probability of the swap test can be found from \cref{eq:prob} and the 
spectrum of $\rho(\e)$:
\begin{equation}
\begin{aligned}
  p(\e)
  &:= \frac{1}{2}
      \of[\bigg]{ 1 + 
        \of[\Big]{ (1-\e) + \frac{\e}{d} }^2 
      + (d-1) \of[\Big]{\frac{\e}{d}}^2
      } \\
  &\:= 1 - \of[\Big]{1 - \frac{1}{d}} \e + \frac{1}{2} \of[\Big]{1 - \frac{1}{d}} \e^2.
  \label{eq:p}
\end{aligned}
\end{equation}
Accordingly, the expected number of copies of $\rho(\e)$ consumed by the $\SWAP$ gadget is
\begin{equation}
	\label{eq:exp_cpy}
  	\frac{2}{p(\e)} = \frac{2d}{d - (d-1) \e + \frac{1}{2}(d-1) \e^2}.
\end{equation}
The output state of the $\SWAP$ gadget is 
$\SWAP \of[\big]{\rho(\e),\rho(\e)} = \rho(\e')$, where $\e'$ can be obtained 
from \cref{eq:evals}:
\begin{equation}
  \frac{\e'}{d} = \lambda_2' = \frac{ \frac{\e}{d} + \left( \frac{\e}{d} \right)^2}{2 p(\e)},
\end{equation}
so 
\begin{equation}
  \e' = \frac{\e + \e^2/d}{2 p(\e)} 
      = \frac{\e + \e^2/d}{2 - 2 \of[\Big]{1 - \frac{1}{d}} \e + \of[\Big]{1 - \frac{1}{d}} \e^2 }.
  \label{eq:e}
\end{equation}

If the inputs are only similar but not identical, the
above analysis holds approximately.  In \cref{sec:io}, we evaluate the
output error parameter of the $\SWAP$ gadget when the inputs are
$\rho(\e_1)$ and $\rho(\e_2)$ for arbitrarily $\e_1,\e_2$.  We present 
the region of $(\e_1,\e_2)$ for which the output has better purity
than both inputs, which includes the line $\e_1=\e_2$.  For large
$\e_1, \e_2$, the region becomes narrower as $d$ increases.  We
thus design the protocol to apply the swap test to states of equal
noise parameter.

\subsection{Recursive purification based on the swap test} \label{sec:recursive}

Given access to many copies of $\rho(\e)$ from \cref{eq:re}, with a given value of $\e \in (0,1)$, our goal is to produce a single copy of $\rho(\epsilon)$,
for some desired target parameter $\epsilon \ll \e$.
Since our purification procedure (\cref{alg:purify}) invokes itself recursively, it will be convenient to denote the initial state by $\rho_0 := \rho(\e)$
and the state produced at the $i$th level of the recursion by $\rho_i$.
We choose the total depth of the recursion later to guarantee a desired level of purity in the final output.

Before stating our algorithm more formally, we emphasize one slightly unusual aspect: even though we fix the total depth of the recursion, the number of calls \emph{within} each level is not known in advance. This is because each instance of the algorithm calls an instance in the previous recursion level an unknown number of times (as many times as necessary for the $\SWAP$ gadget to succeed). We analyze the expected number of calls and the total number of input states consumed in \cref{sec:complexity}.

\vspace*{1ex}

\begin{algorithm}[h]
\caption{Recursive purification based on the swap test.}
\label{alg:purify}
\emph{This procedure purifies a stream of states $\rho_0$ by recursively calling the $\SWAP$ gadget (\cref{alg:swap}).}
\begin{algorithmic}[1]
\Procedure{$\PURIFY$}{$n$}
\If{$n=0$}
\Comment{If at the bottom level of recursion,}
  \State \Return $\rho_0$
  \Comment{request one copy of $\rho_0$ from the stream and return it.}
\Else
\Comment{Otherwise call the next level of recursion}
  \State \Return $\SWAP(\PURIFY(n-1),\PURIFY(n-1))$
  \Comment{until the $\SWAP$ gadget succeeds.}
\EndIf
\EndProcedure
\end{algorithmic}
\end{algorithm}

\vspace*{1ex}

Let $\rho_n := \PURIFY(n)$ denote the output of a purification procedure of depth $n$. We can represent this procedure as a full binary tree of height $n$ whose $2^n$ leaves correspond to the original input states $\rho_0 = \rho(\e)$ and whose root corresponds to the final output state $\rho_n$.
For any $i \in \set{0,\dotsc,n}$, there are $2^{n-i}$ calls to $\PURIFY(i)$, each producing a copy of $\rho_i$ by combining two states from the previous level (or by directly consuming one copy of $\rho_0$ when $i=0$).
Every time a $\SWAP$ gadget fails, the recursion is restarted at that level.
This triggers a cascade of calls going down all the way to $\PURIFY(0)$,
which requests a fresh copy of $\rho_0$.
Hence, a failure at level $i$ results in at least $2^i$ fresh copies of $\rho_0$ being requested.
The procedure terminates once the outermost $\SWAP$ gadget succeeds.

According to \cref{cor:better}, the state $\rho_i$ is of the form $\rho(\e_i)$ for some error parameter $\e_i < \e_{i-1}$. Denote by $p_i$ the probability of success of the $\SWAP$ gadget in $\PURIFY(i)$ (in other words, $p_i$ is the probability of getting outcome $a = 0$ when the swap test is applied on two copies of $\rho_{i-1}$). According to \cref{eq:p,eq:e}, the parameter $p_i$ depends on $\e_i$, while $\e_i$ in turn is found via a recurrence relation:
\vspace*{-0.5ex}
\begin{align}
  p_i &= P(\e_{i-1}, d), &
  \e_i &= \D(\e_{i-1}, d), &
  \e_0 &= \delta,
  \label{eq:recurrence}
\end{align}
where the functions $P, \D\colon (0,1) \times [2,\infty) \to (0,1)$ are defined as follows:
\vspace*{-0.5ex}
\begin{align}
  P(\e, d) &:= 1 - \of[\Big]{1 - \frac{1}{d}} \e + \frac{1}{2} \of[\Big]{1 - \frac{1}{d}} \e^2, &
  \D(\e, d) &:= \frac{\e + \e^2/d}{2 P(\e, d)}.
  \label{eq:PD}
\end{align}
For any values of the initial purity $\e_0 \in (0,1)$ and dimension $d \geq 2$, these recursive relations define a pair of sequences $(p_1, p_2, \dotsc)$ and $(\e_0, \e_1, \dotsc)$.
Our goal is to understand the asymptotic scaling of the parameters $p_i$ and $\e_i$, and to determine a value $n$ such that $\e_n \leq \epsilon$, for some given $\epsilon > 0$.

\subsection{Non-recursive stack-based implementation}\label{sec:stackbased}

It is not immediately obvious from the recursive description of $\PURIFY(n)$ in \cref{alg:purify} that it uses $n+1$ qudits of quantum memory plus one ancilla qubit.
\cref{alg:stackbased purify} provides an alternative stack-based implementation of the same procedure that makes this clear.


\begin{algorithm}[ht]
\newcommand{\bul}{$\triangleright$~~}
\newcommand{\qudit}{\mathrm{qudit}}
\newcommand{\purity}{\mathrm{purity}}
\newcommand{\fetch}{FetchNewCopy}
\caption{Stack-based implementation of \cref{alg:purify}.}
\label{alg:stackbased purify}
\vspace*{0.5ex}
\emph{Global variables:}
\vspace*{-0.6ex}
\begin{algorithmic}[1]
  \Statex \bul \emph{$\qudit[\cdot]$ -- an array of qudits that simulates a stack},
  \Statex \bul \emph{$\purity[\cdot]$ -- an array of integers that store the purity level of each qudit},
  \Statex \bul \emph{$k$ -- an integer that indicates the current position in the stack}.
\end{algorithmic}
\emph{This procedure requests a fresh copy of $\rho_0$ from the stream and stores it in the stack.}
\begin{algorithmic}[1]
\Procedure{\fetch}{}
  \State $k$++
  \Comment{Move the stack pointer forward by one.}
  \State $\qudit[k] = \rho_0$
  \Comment{Fetch a fresh copy of $\rho_0$ and store it in the current stack position.}
  \State $\purity[k] = 0$
  \Comment{Set its purity level to $0$.}
\EndProcedure
\end{algorithmic}
\emph{Non-recursive implementation of the $\PURIFY(n)$ procedure from \cref{alg:purify}.}
\begin{algorithmic}[1]
\Procedure{Purify}{$n$}
  \State $k = 0$
  \Comment{Initially the stack is empty.}
  \State $\purity[0] = -1$
  \Comment{Set the purity level of empty stack to a non-existing value.}
  \Repeat
  \Comment{Repeat until $n$ levels of recursion are reached.}
    \State \Call{\fetch}{}
    \Comment{Request two fresh copies of $\rho_0$}
    \State \Call{\fetch}{}
    \Comment{and put them in the stack.}
    \Repeat
    \Comment{Keep merging the two top-most qudits.}
      \State Apply the swap test circuit shown in \cref{fig:swap} to $\qudit[k-1]$ and $\qudit[k]$,
      \Statex store the result in $\qudit[k-1]$, and let $a$ be the measurement outcome.
      \If{$a = 0$}
      \Comment{If the swap test succeeded,}
      \State $k = k-1$
      \Comment{move the pointer backwards by one (i.e., keep the improved qudit)}
      \State $\purity[k]$++
      \Comment{and increase its purity level by one.}
      \Else
      \Comment{If the swap test failed,}
      \State $k = k-2$
      \Comment{move the pointer backwards by two (i.e., discard both qudits).}
      \EndIf
    \Until{$a=1$ or $\purity[k-1] \neq \purity[k]$}
    \Comment{Stop if swap test failed or different purities.}
  \Until{$\purity[1] = n$}
  \Comment{Stop when $n$ levels of recursion are reached.}
  \State \Return{$\qudit[1]$}
  \Comment{Return the only remaining qudit.}
\EndProcedure
\end{algorithmic}
\end{algorithm}

\vspace*{-2ex}

\clearpage 

Note that throughout the execution of \cref{alg:stackbased purify}, $k \geq 0$, purity$[0] = -1$, and 
purity$[1] \geq$ purity$[2] \geq \cdots \geq$ purity$[k]$ with at most one equality. 
The last statement can be verified by noting that it is satisfied initially, and that it is preserved whenever any of the purity values are updated.
The fact that at most two purities can be equal implies the stated upper bound for the memory required.
Just as in the recursive formulation in \cref{alg:purify}, the swap test is always applied to 
two states of equal purity. In the first pass through the repeat loop in line 7, the swap test is applied to two new 
copies of $\rho_0$. In subsequent passes, the condition in line 14 ensures that the next inputs to the swap test have equal purity.

\section{Analysis of recursive purification}\label{sec:analysis}

The goal of this section is to analyze the complexity of the recursive
purification procedure $\PURIFY(n)$ defined in \cref{alg:purify}.
Most of the analysis deals with understanding the recurrence
relations in \cref{eq:recurrence,eq:PD}.
First, in \cref{sec:monotonicity} we establish monotonicity of the
functions $P$ and $\D$ from \cref{eq:PD} and the corresponding parameters $p_i$ and $\e_i$.
In \cref{sec:asymptotics}, we analyze the asymptotic scaling of the error parameter $\e_i$.
Finally, in \cref{sec:complexity}, we analyze the expected sample complexity of $\PURIFY(n)$.

\subsection{Monotonicity relations}\label{sec:monotonicity}

First we prove monotonicity of the functions $P$ and $\D$ in \cref{eq:PD}.

\begin{lemma}\label{claim:monotonicity}
For any $\e \in (0,1)$ and $d \geq 2$,
the function $P(\e,d)$ is strictly decreasing in both variables, and  
the function $\D(\e,d)$ is strictly increasing in both variables.
\end{lemma}

\begin{proof}
To verify the monotonicity of $P$, note that
\begin{align}
  \frac{\partial P}{\partial \e}
  &= \of[\Big]{1 - \frac{1}{d}} \of{-1 + \e} < 0, &
  \frac{\partial P}{\partial d}
  &= \of[\Big]{-1 + \frac{\e}{2}} \frac{\e}{d^2} < 0.
\end{align}
Similarly, the partial derivatives of $\D$ satisfy
\begin{align}
  \frac{\partial \D}{\partial \e}
  &= \frac{1+2\e/d}{2P(\e,d)}
   - \frac{\e+\e^2/d}{2P(\e,d)^2} \cdot
     \frac{\partial}{\partial \e} P(\e,d)
   = \frac{(d + 2\e) P(\e,d) + \e (d+\e) (1-1/d) (1-\e)}{2 d P(\e,d)^2} > 0, \\
  \frac{\partial \D}{\partial d}
  &= \frac{-\e^2/d^2}{2 P(\e,d)}
   - \frac{\e+\e^2/d}{2 P(\e,d)^2} \cdot
     \frac{\partial}{\partial d} P(\e,d)
   = \frac{(1-\e)\e^3}{4 d^2 P(\e,d)^2} > 0,
\end{align}
which completes the proof.
\end{proof}

Thanks to the recurrence relations in \cref{eq:recurrence},
the parameters $p_i$ and $\e_i$ inherit the monotonicity of $P$ and $\D$ from \cref{claim:monotonicity}.  We denote them by $p_i^{(d)}$ and $\e_i^{(d)}$ to explicitly indicate their dependence on $d$.
The following lemma shows that both parameters are strictly monotonic in the dimension $d$.

\begin{lemma}\label{claim:monind}
For any $i \geq 1$,
the success probability $p_i^{(d)}$ is monotonically strictly decreasing in $d$ and
the error parameter $\e_i^{(d)}$ is monotonically strictly increasing in $d$,
for all integer values of $d \geq 2$.
\end{lemma}

\begin{proof}
Recall from \cref{eq:recurrence} that $\e_0^{(d)} = \e = \e_0^{(d+1)}$, irrespective of $d$.
Using the recurrence from \cref{eq:recurrence} and the strict monotonicity of $\D$ in $d$ from \cref{claim:monotonicity},
\begin{equation}
  \e_1^{(d)}
  = \D(\e_0^{(d)}, d)
  < \D(\e_0^{(d+1)}, d+1)
  = \e_1^{(d+1)}.
\end{equation}
Using this as a basis for induction,
the strict monotonicity of $\D$ in both arguments implies that
\begin{equation}
  \e_i^{(d)}
  = \D(\e_{i-1}^{(d)}, d)
  < \D(\e_{i-1}^{(d+1)}, d+1)
  = \e_i^{(d+1)}
  \label{eq:emon}
\end{equation}
for all $i \geq 1$.
We can promote this to a similar inequality for $p_i^{(d)}$ using the recurrence from \cref{eq:recurrence}.
Indeed, this recurrence together with the strict monotonicity of $P$ in both arguments implies that
\begin{equation}
  p_i^{(d)}
  = P(\e_{i-1}^{(d)}, d)
  > P(\e_{i-1}^{(d+1)}, d+1)
  = p_i^{(d+1)}
  \label{eq:pmon}
\end{equation}
for all $i \geq 1$.
\end{proof}

\subsection{Asymptotics of the error parameter}\label{sec:asymptotics}

Our goal is to upper bound the error sequence $\e_i^{(d)}$ as a function of
$\e = \e_0^{(d)}$, $d$, and $i$.  We also define 
\begin{equation}
  \e_i^{(\infty)} := \lim_{d \to \infty} \e_i^{(d)}
  \label{eq:einf}
\end{equation}
for all $i \geq 0$.
We can gain some insight and motivate our analysis by examining
numerical plots of $\e_i^{(d)}$ as functions of $i$, for some
representative values of $d$ and a common $\e_0^{(d)}$ that is very
close to $1$, as shown in \cref{fig:numericalplotsfordelta}.
To understand the error sequence with a
smaller value of $\e_0^{(d)}$, one can simply translate the plot horizontally, starting at a value of $i$ corresponding to the desired initial error.

\begin{figure}[H]
\centering
\includegraphics[width = 0.7\textwidth]{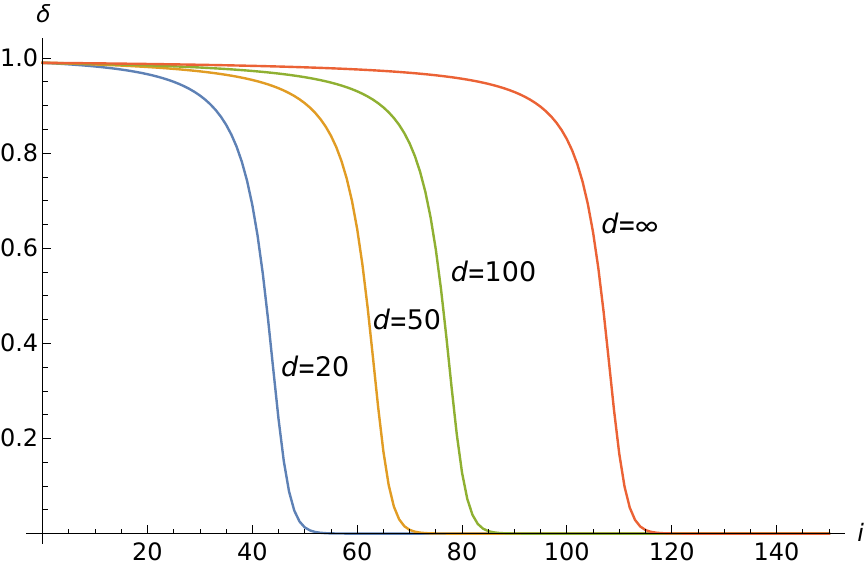}
\caption{\label{fig:numericalplotsfordelta}Numerical plots of $\e_i^{(d)}$
for $d = 20, 50, 100$, and of $\e_i^{(\infty)}$, with initial error $\e_0^{(d)} = 0.99$.} 
\end{figure}

For each $d$, the sequence $\e_{i}^{(d)}$ features a slow decline 
from the initial value $\e \approx 1$, until it drops to about $0.9$.  Then the sequence drops rapidly to about $0.05$ before leveling off and decreasing exponentially. 
We can establish asymptotic bounds by separately considering
the subsequences with $\e_{i}^{(d)} \geq 1/2$ and $\e_{i}^{(d)} < 1/2$.  
However, tighter bounds can be obtained by avoiding the threshold
$\e_{i}^{(d)} \approx 1/2$.  
Thus we consider the subsequences with $\e_{i}^{(d)} \leq 1/3$, 
$1/3 < \e_{i}^{(d)} \leq 2/3$, and $2/3 < \e_{i}^{(d)}$.
Depending on how small $\e$ is, we can combine some or all of the
above estimates to obtain the number of iterations needed to suppress
the error from $\e$ to $\epsilon$ for arbitrary $d$. We then use these results to bound the sample complexity in \cref{sec:complexity}.
Note that the thresholds $1/3$ and $2/3$ are arbitrary, and can be
further optimized as a function of $d$, but we leave such fine tuning to the interested reader.

Our analysis exploits the monotonicity established in \cref{sec:monotonicity}: for a common $\e_0^{(d)}$, for all $i,d$, we have $\e_{i}^{(d)} \leq
\e_{i}^{(d+1)} \leq \e_i^{(\infty)}$, since $\D(\e_{i}^{(d)}, d)$ is an
increasing function of $d$ and of $\e_{i}^{(d)}$.
We divide our analysis into parts corresponding to the three subsequences:
\begin{enumerate}
    \item \label{part:small} We establish a tight upper bound for $\e_i^{(d)}$ for 
    $\e_0^{(d)} = \e \leq 1/3$ for all $d \geq 2$.  This bound is 
    roughly exponentially decreasing in $i$.
    \item \label{part:medium} For any $d \geq 2$, we show that $3$ to $5$ iterations are enough to reduce $\e_i^{(d)}$
    from $2/3$ to $1/3$.
    \item \label{part:large} For $\e_0^{(d)} > 2/3$, and specifically for 
    $\e_0^{(d)} \approx 1$, we are interested in upper bounds on $i^*$ so
    that $\e_{i^*}^{(d)} \approx 2/3$.  We further divide this step into 
    two parts.
    \begin{enumerate}[label=(\alph*),ref=\theenumi(\alph*)]
        \item \label{part:large-a} First, we upper bound 
$i^*$ for $\e_0^{(\infty)}$.  Due to \cref{claim:monind}, the 
same upper bound holds for other $\e_0^{(d)}$.  
Combining parts \ref{part:small}, \ref{part:medium}, and \ref{part:large-a} (in \cref{sec:part1}, \cref{sec:part2}, and \cref{sec:part3a}, respectively), 
we can upper bound $\e_{i}^{(d)}$ for all
$\e$, $i$, and $d$. In particular, we estimate the number of
iterations required (as a function of $\e=\e_0$ and $\epsilon$) to
purify copies of $\rho(\e)$ into $\rho(\epsilon)$.  
The number of iterations is a gentle function of $\e$ and
$\epsilon$, and the large-$d$ limit only introduces an additive constant.  

\item \label{part:large-b} Then we tighten the estimate of $i^*$
for $\e_{i}^{(d)}$ with techniques similar to those in \cref{part:large-a}, but
with more complex arithmetic.  The resulting estimate applies to
all $d \geq 2$ and $\e>2/3$.  It improves upon the bound in
\cref{part:large-a} for small $d (1-\e)$, and approaches the
bound in \cref{part:large-a} for large $d (1-\e)$.

\end{enumerate}
\end{enumerate}

\subsubsection{\Cref{part:small}} 
\label{sec:part1}

If the initial error parameter $\e_0^{(d)}$ is sufficiently small, subsequent values $\e_i^{(d)}$ decrease roughly exponentially in $i$.

\begin{lemma}\label{claim:bounds}
For all $d \geq 2$, if $\e = \e_0^{(d)} \leq 1/2$, then for all $i \geq 0$,
\begin{align}
  \e_i^{(d)} &\leq \frac{\e}{2^i (1 - 2 \e) + 2 \e} .
  \label{eq:012}
\end{align}
\end{lemma}

\begin{proof}
We upper bound $\e_i^{(d)}$ by another sequence $\eta_i$ defined as follows:
\begin{align}
  \eta_{0} := \delta, \quad
  \eta_{i} := \frac{\eta_{i-1}}{2 - 2 \eta_{i-1}} \quad {\rm for}~i \geq 1.
  \label{eq:eta}
\end{align}
To see that $\e_i^{(d)} \leq \eta_i$ for all $i \geq 0$, suppose $\e =
\e_0^{(d)} = \e_0^{(\infty)} = \eta_{0}$.
From \cref{claim:monind}, $\e_i^{(d)} \leq \e_i^{(\infty)}$ for all $i
\geq 0$.  Let us now prove by induction that $\e_i^{(\infty)} \leq
\eta_i$ for all $i \geq 0$.  The base case is immediate.  The induction hypothesis says
$\e_{r}^{(\infty)} \leq \eta_{r}$ for some $r\geq 0$, so from \cref{eq:recurrence}, we have 
\begin{align}
    \e_{r+1}^{(\infty)}
    & = \D(\e_{r}^{(\infty)}, \infty)
    \leq \D(\eta_{r}, \infty)
\\
    & = \frac{\eta_{r}}{2 - 2 \eta_{r} + \eta_{r}^2}
    \leq \frac{\eta_{r}}{2 - 2 \eta_{r}} = \eta_{r+1}
\end{align}
where the second line comes from \cref{eq:e}.
Altogether, this shows $\e_i^{(d)} \leq \eta_i$ for all $i \geq 0$.  

The recurrence $\eta_i$ admits an exact solution.
Indeed, note that
\begin{equation}
  \frac{1}{\eta_i} = \frac{2}{\eta_{i-1}} - 2,
\end{equation}
which yields the following closed-form expression for $\eta_i$:
\begin{equation}
  \frac{1}{\eta_i} = 2^i \of*{\frac{1}{\eta_0} - 2} + 2.
\end{equation}
Inverting and substituting $\eta_0 = \e$, we find that
\begin{equation}
  \eta_i = \frac{\e}{2^i (1 - 2 \e) + 2 \e}.
  \label{eq:etai}
\end{equation}
and the result follows.
\end{proof}

Note that under the assumption $\e \leq 1/2$, we have $0 \leq \eta_i \leq 1$ for all $i \geq 0$.

To suppress the error from $\e \in (0,1/2)$ to
$\epsilon$, since $\eta_i < \frac{\e}{2^i (1 - 2 \e)}$, it suffices to
iterate $\log_2 \frac{\e}{(1-2\e) \epsilon}$ times.  In particular, for
$\e \leq 1/3$, it suffices to iterate $\log_2 \frac{1}{\epsilon}$
times.

\subsubsection{\Cref{part:medium}}
\label{sec:part2}

Using monotonicity from \cref{claim:monind} and by considering $\e^{(\infty)}_i$, a brute-force calculation shows that $5$ iterations are sufficient to suppress the error from at most $2/3$ to at most $1/3$ for all $d \geq 2$.  For smaller values of $d$, as few as $3$ iterations suffice.

\subsubsection{\Cref{part:large-a}}
\label{sec:part3a}

For $\e \approx 1$, \cref{fig:numericalplotsfordelta} suggests that the error suppression is slow until after about $\frac{1}{1-\e}$ iterations, when the error rapidly drops from $0.9$ to $0.1$.  We now
study how quickly $\e_i^{(d)}$ can be upper bounded by $2/3$ starting
from an arbitrary $\e_0^{(d)} \approx 1$.

We first consider the simpler sequence  $\e_i^{(\infty)}$.
Note from \cref{eq:PD} that
\begin{equation}
    \lim_{d \to \infty} \D(x,d) = \frac{x}{2 - 2x + x^2},
\end{equation}
so $\e_i^{(\infty)}$ obeys the following recurrence:
\begin{align}
  \e_0^{(\infty)} := \delta, \quad
  \e_i^{(\infty)} := \frac{\e_{i-1}^{(\infty)}}{ 2 - 2 \e_{i-1}^{(\infty)} + \e_{i-1}^{(\infty)2}} \quad {\rm for}~i \geq 1.
\end{align}
Because the function $(2-2x+x^2)^{-1}$ is monotonically increasing for $x \in [0,1]$, and $\e_i^{(\infty)}$ is monotonically decreasing in $i$, the following upper bound holds by induction:
\begin{equation}
  \e_i^{(\infty)}
  \leq \frac{\e_{i-1}^{(\infty)}}{2 - 2 \e + \e^2}
  \leq \frac{\e}{ (2 - 2 \e + \e^2)^i} .
\label{eq:older-delta-bound} 
\end{equation}
If $\e > 2/3$ is constant, 
the above is an exponentially decreasing function in $i$, which may be 
sufficient for many purposes.  
However, this bound does not have tight dependence on $\e$.  

We can instead bound $\e_i^{(\infty)}$ by considering two other recurrences, defined as follows.  
Let $i^*$ be the smallest $i$ so that $\e_{i^*+1}^{(\infty)} < 2/3$. 
We first define 
\begin{equation}
\kappa_i^{(\infty)} = 1-\e_i^{(\infty)},
\end{equation}
so that
\begin{equation}
\label{eq:kappai_expression}
\kappa_0^{(\infty)} 
= 1-\e, \quad
\kappa_i^{(\infty)} = 1 - \e_i^{(\infty)} 
= 1 - \frac{\e_{i-1}^{(\infty)}}{ 2 - 2 \e_{i-1}^{(\infty)} + \e_{i-1}^{(\infty)2}}
= \frac{\kappa_{i-1}^{(\infty)} + \kappa_{i-1}^{(\infty)2}}{1+ \kappa_{i-1}^{(\infty)2}}.
\end{equation}
Let
\begin{equation}
g^{(\infty)}(x) = \frac{x+x^2}{1+x^2},
\label{eq:ginfinity}
\end{equation}
so $\kappa_i^{(\infty)} = g^{(\infty)}(\kappa_{i-1}^{(\infty)})$.  
Furthermore, for $i \in \{0,1,\ldots, i^*\}$, define an inverse recurrence for $\kappa_i^{(\infty)}$ as 
\begin{equation}
\mu_0^{(\infty)} 
= \kappa_{i^*}^{(\infty)}, \quad
\mu_i^{(\infty)} = h^{(\infty)}(\mu_{i-1}^{(\infty)})
\end{equation}
where $h^{(\infty)}(y)$ is the inverse of $y = g^{(\infty)}(x) = \frac{x+x^2}{1+x^2}$ (which is well-defined for all $x \in [0,1]$), 
so indeed, $\mu_i^{(\infty)} = \kappa_{i^*-i}^{(\infty)}$.
In particular, $\mu_{i^*}^{(\infty)} = 1-\e$.  
The definition of $i^*$ and the relationships between 
$\e_i^{(d)}, \kappa_i^{(d)}, \mu_i^{(d)}$ for a representative
$d=20$ are summarized in \Cref{fig:kappa-mu}.

\begin{figure}[ht]
\vspace{3ex}
\begin{center}
  \includegraphics[width = .8\textwidth]{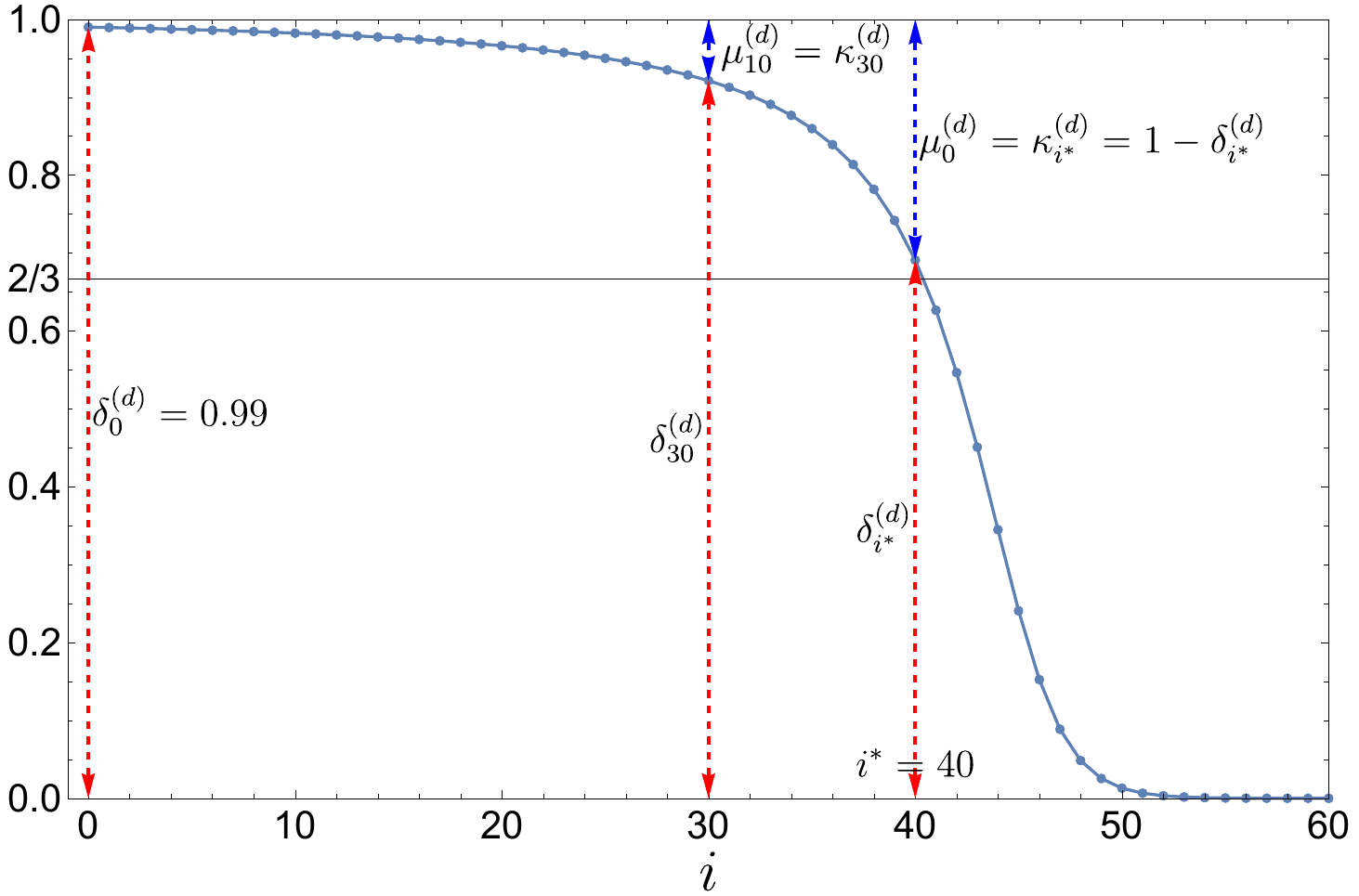}
\end{center}
\caption{\label{fig:kappa-mu} A diagrammatic summary of the definition of $i^*$, 
and the relationships between $\e_i^{(d)}, \kappa_i^{(d)}$, and $\mu_{i^*-i}^{(d)}$, illustrated for $d=20$ and $\e=0.99$. 
The light blue curve shows $\e_i^{(d)}$ as a function of $i$.
Since $i^*$ is the smallest $i$ so that $\e_{i+1}^{\infty} < 2/3$, in this case, $i^*=40$.
For each $i$, $\kappa_i^{(d)} = 1- \e_i^{(d)}$ is the distance from the 
light blue curve to the value $1$.  Finally, $\mu_{i}^{(d)} = \kappa_{i^*-i}$ 
is a recurrence that iterates from the right to the left (starting from 
$i^*$).  Note that the 
index $j$ for $\mu_j^{(d)}$ is related to the label $i$ for the abscissa as $j+i = i^*$.} 
\end{figure}

\clearpage

\begin{figure}[ht]
\begin{center}
\includegraphics[width = .8\textwidth]{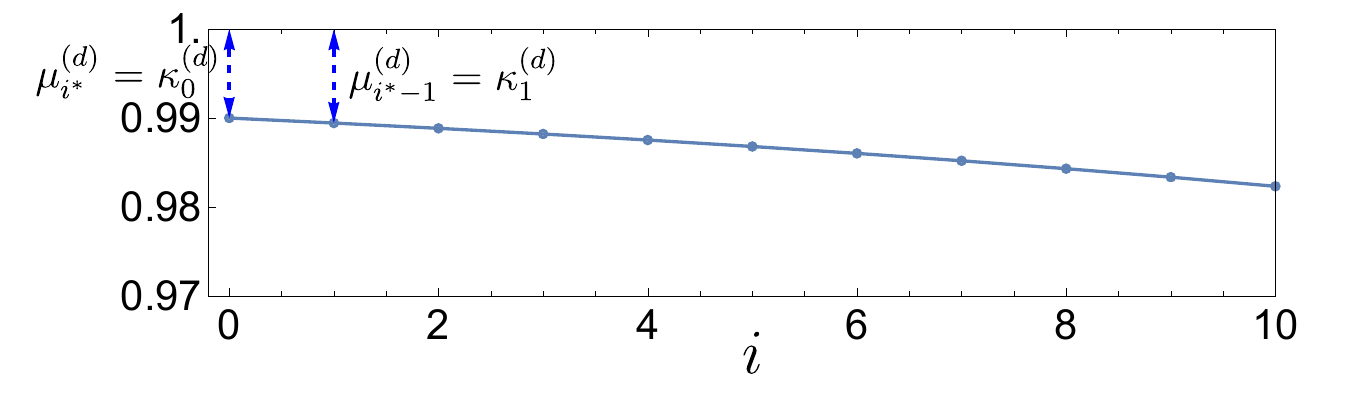}
\end{center}
\caption{
Zoomed-in view of the upper left corner of \Cref{fig:kappa-mu}.}
\end{figure}


The main idea of the analysis is as follows.  
We will derive upper bounds for $\mu_i^{(\infty)}$ as a function of $i$. 
In particular, 
if we prove an upper bound $\upsilon$ for $\mu_{i}^{(\infty)}$,
we have a lower bound $1-\upsilon$ for $\e_{i^*-i}^{(\infty)}$, which means
that $i+1$ iterations are sufficient to reduce 
$\e_{i^*-i}^{(\infty)}$ from at least $1-\upsilon$ to
$\e_{i^*+1}^{(\infty)} < 2/3$.  
In particular, if $n$ satisfies $\mu_n^{(\infty)} \leq 1-\e$, then
$n+1$ iterations are sufficient to reduce the initial error
$\e_{0}^{(\infty)} = \e$ to $\e_{n+1}^{(\infty)} < 2/3$, so $i^*<n$.
In particular, we show the following.\footnote{In version 1 of the arXiv preprint of this
manuscript, we invoked a theorem from \cite{Stevic96} (case 2), which
implies
\begin{equation*}
    \mu_i^{(\infty)} = \frac{1}{i} + \frac{2 \ln i}{i^2} + o \left( \frac{\ln i}{i^2} \right), 
\end{equation*}
and provided a self-contained derivation using Taylor series expansion
in $\mu_i^{(\infty)}$.  The strict bound in
\cref{claim:mu-inf-bound} is obtained by different techniques, 
most crucially by strengthening Eq.~(54) in version 1 to 
the inequality in \cref{claim2}, and by replacing the approximation $\mu_i^{(\infty)} \approx \frac{1}{i}$ with recursively better upper bounds in \cref{claim3,claim4,claim5}, using \cref{claim2}.\label{footnote:v1comparison}}

\begin{lemma}\label{claim:mu-inf-bound}
If $\mu_0^{(\infty)} \leq 1/3$, then for all $i \geq 1$, 
\begin{align}
\mu_i^{(\infty)} < \frac{1}{i} + \frac{2 \ln i}{i^2}.
\end{align}
\end{lemma}


\begin{proof}
The proof consists of several claims.  In \cref{claim2}, we obtain an
inequality which can be used to upper bound $\mu_i^{(\infty)}$ 
using an upper bound for $\mu_{i-1}^{(\infty)}$.  Then, we start with a crude upper bound on all $\mu_k^{(\infty)}$ in \cref{claim1}, and recursively improve the upper bound in \cref{claim3,claim4,claim5}, using \cref{claim2}.

By taking the derivative of $g^{(\infty)}(x)$ defined
in \cref{eq:ginfinity}, we can verify that
$g^{(\infty)}(x)$ is \emph{strictly} increasing, continuous,
and bounded on $[0,2.4]$.  
Therefore, its inverse $h^{(\infty)}(y)$ has the same properties
on $[0,1.2]$. 
We can also obtain a closed-form expression for $h^{(\infty)}(y)$
by inverting $y = \frac{x+x^2}{1+x^2}$ 
using the quadratic formula, 
\begin{equation} 
x = \frac{-1 + \sqrt{1 + 4 y (1-y)}}{2 (1-y)}, 
\end{equation}
keeping only the solution corresponding to 
the case of interest, $0 < x,y < 1/3$
(because $1/3 > \mu_0^{(\infty)}$ and $\mu_i^{(\infty)}$ 
decreases with $i$ while remaining positive). 

\begin{claim}\label{claim1}
We have
$\mu_1^{(\infty)} \leq 0.2808$,
$\mu_2^{(\infty)} \leq 0.2396$, and $\mu_i^{(\infty)} \leq 1/10$ for all $i \ge 10$. 
\end{claim}

\begin{claimproof}[Proof of \cref{claim1}]
This follows from the monotonicity of $h^{(\infty)}$
and its closed-form expression,  
using $\mu_0^{(\infty)} \leq 1/3$. \vspace*{3ex}
\end{claimproof}

\begin{claim}\label{claim2}
For all $i \geq 3$,  
\begin{equation}
\frac{1}{\mu_i^{(\infty)}} > \frac{1}{\mu_{i-1}^{(\infty)}} + 1 - 2 \mu_{i-1}^{(\infty)} .
\label{eq:claim1}
\end{equation}
\end{claim}

\begin{claimproof}[Proof of \cref{claim2}]
We can rephrase the claim as 
\begin{equation}
\mu_i^{(\infty)} = h^{(\infty)}(\mu_{i-1}^{(\infty)}) < \frac{1}{\frac{1}{\mu_{i-1}^{(\infty)}} + 1 - 2 \mu_{i-1}^{(\infty)}} . 
\end{equation}
For $i \geq 3$, $\frac{1}{\frac{1}{\mu_{i-1}^{(\infty)}} + 1 - 2 \mu_{i-1}^{(\infty)}} < 0.22$, so
$g^{(\infty)}$ is invertible. Inserting the identity function $h^{(\infty)} \circ g^{(\infty)}$ on the right-hand side
and using the monotonicity of $h^{(\infty)}$, it suffices to show
\begin{equation}
\mu_{i-1}^{(\infty)} < g^{(\infty)} \left( \frac{1}{\frac{1}{\mu_{i-1}^{(\infty)}} + 1 - 2 \mu_{i-1}^{(\infty)}} \right).
\label{eq:newcondclaim2}
\end{equation}
The above steps translate the analysis from $h^{(\infty)}$ (involving
a square root) to $g^{(\infty)}$ (which involves simpler rational expressions and squares).  

For $x>0$, $g^{(\infty)}(x) = \frac{x+x^2}{1+x^2}
=\frac{1+\frac{1}{x}}{1+\frac{1}{x^2}}$, so \cref{eq:newcondclaim2}
is equivalent to
\begin{equation}
 0 <
 \frac
     {1 + \left( \frac{1}{\mu_{i-1}^{(\infty)}} + 1 - 2 \mu_{i-1}^{(\infty)} \right)^{~}}
     {1 + \left( \frac{1}{\mu_{i-1}^{(\infty)}} + 1 - 2 \mu_{i-1}^{(\infty)} \right)^2}
 - \mu_{i-1}^{(\infty)},
\end{equation}
or equivalently
\begin{align}
  0 < 
  \left( \frac{1}{\mu_{i-1}^{(\infty)}} + 2 - 2 \mu_{i-1}^{(\infty)} \right)
  - 
  \mu_{i-1}^{(\infty)} - \mu_{i-1}^{(\infty)}
  \left( \frac{1}{\mu_{i-1}^{(\infty)}} + 1 - 2 \mu_{i-1}^{(\infty)} \right)^2 
  = 4 \mu_{i-1}^{(\infty)2} \left( 1- \mu_{i-1}^{(\infty)}\right),
\end{align}
which clearly holds.
\end{claimproof}

Note that \cref{claim2} shows an upper bound on $\mu_{i}^{(\infty)}$ using any upper bound on $\mu_{i-1}^{(\infty)}$.  We use this strategy to recursively establish tight bounds in the following claims.

\begin{claim}\label{claim3}
For all $i \geq 3$,  
\begin{equation}
 \frac{1}{\mu_i^{(\infty)}} > \frac{4}{5} \left( i + \frac{5}{2} \right) .
\label{eq:claim3}
\end{equation}
\end{claim}

\begin{claimproof}[Proof of \cref{claim3}]
The claim can be directly verified 
for $1 \leq i \leq 10$ using the upper bound $\mu_0^{(\infty)} \leq 1/3$ and 
the closed-form expression for $h^{(\infty)}$.  For $i \geq 11$, using 
\cref{claim1}, $1- 2 \mu_{i-1}^{(\infty)} \geq \frac{4}{5}$.  Substituting this
into \cref{claim2} gives
\begin{equation}
 \frac{1}{\mu_i^{(\infty)}} 
> \frac{1}{\mu_{i{-}1}^{(\infty)}} + \frac{4}{5}
> \frac{1}{\mu_{i{-}2}^{(\infty)}} + 2 \times \frac{4}{5}
> \cdots 
> \frac{1}{\mu_{10}^{(\infty)}} + (i-10) \times \frac{4}{5}
> \frac{4}{5} \left( i + \frac{5}{2} \right),
\end{equation}
where we have also used \cref{claim1} for $\mu_{10}^{(\infty)}$ in the last step. 
\end{claimproof}

\Cref{claim3} improves the bound in \cref{claim1}. We can again 
use \cref{claim2}, now together with \cref{claim3}, to get another improved bound.  

\begin{claim}\label{claim4} 
For all $i \geq 1$,  
\begin{equation}
 \mu_i^{(\infty)} < \frac{1}{i} + \frac{5}{2} \frac{\ln i}{i^2} .
\label{eq:claim4}
\end{equation}
\end{claim}

\begin{claimproof}[Proof of \cref{claim4}]
The $i=1,2$ cases follow from \cref{claim1}.  
Using \cref{claim2,claim3}, for all $i \geq 3$,   
\begin{equation}
\frac{1}{\mu_i^{(\infty)}} 
> \frac{1}{\mu_{i-1}^{(\infty)}} + 1 - 2 \mu_{i-1}^{(\infty)} 
> \frac{1}{\mu_{i-1}^{(\infty)}} + 1 - \frac{5}{2} \frac{1}{\left( i + \frac{5}{2} \right) }
> \frac{1}{\mu_{2}^{(\infty)}} + (i-2) - \frac{5}{2} \sum_{k=2}^{i-1} 
                                     \frac{1}{\left( k + \frac{5}{2} \right) } .
\label{eq:claim4step1}
\end{equation}
Note that 
\begin{equation}
\frac{1}{x + \frac{1}{2}} \leq \ln(x+1) - \ln(x) \text{~for $x \geq 1$},
\label{eq:lnbound}
\end{equation}
so
\begin{align}
\frac{1}{\mu_i^{(\infty)}} 
& > \frac{1}{\mu_{2}^{(\infty)}} + (i-2) - \frac{5}{2} \sum_{k=2}^{i-1} 
                                     ~ \left[ \ln(k+3) - \ln(k+2) \right] 
\\
& = \frac{1}{\mu_{2}^{(\infty)}} + (i-2) - \frac{5}{2} 
                                     ~ \left[ \ln(i+2) - \ln 4 \right] 
\\
& > i - \frac{5}{2} \ln(i+2) + 5.6408 ,
\\
& >i-\frac{5}{2}\ln i+4, 
\label{eq:claim4step2}
\end{align}
where we use the numerical bound of $\mu_{2}^{(\infty)}$ and 
$\ln(i+2)-\ln i=\ln(1+\frac{2}{i})\leq \ln\frac{5}{3}$ in 
the last two steps.  
It remains to show
\begin{equation}
    \frac{1}{  i - \frac{5}{2} \ln i + 4}< \frac{1}{i} +  \frac{5}{2} \frac{\ln i}{i^2},
\end{equation}
which is equivalent to 
\begin{equation}
    16 i + 40\ln i> 25\ln^2i.
\end{equation}
This holds since $16 x> 25\ln^2x$ for all $x\geq 1$. Indeed, the function $4\sqrt{x}-5\ln x$ is minimized at $x=\frac{25}{4}$ where it is still positive.
\end{claimproof}

We provide the final tightening of 
the bound that recovers the asymptotic result mentioned in \cref{footnote:v1comparison} and completes the proof of \cref{claim:mu-inf-bound}.

\begin{claim}\label{claim5}
For all $i \geq 1$,  
\begin{equation}
 \mu_i^{(\infty)} < \frac{1}{i} + \frac{2\ln i}{i^2} .
\label{eq:claim5'}
\end{equation}
\end{claim}

\begin{claimproof}[Proof of \cref{claim5}]
The case of $i\leq 7$ can be checked directly. For all $i\geq 8$, \cref{claim2,claim4} imply
\vspace*{-2ex}
\begin{align}
\frac{1}{\mu_i^{(\infty)}} 
& > \frac{1}{\mu_{i-1}^{(\infty)}} + 1 - 2 \mu_{i-1}^{(\infty)} 
> \frac{1}{\mu_{i-1}^{(\infty)}} + 1 - \frac{2}{i-1} -  \frac{5 \ln (i-1)}{(i-1)^2} 
\\
& >  \frac{1}{\mu_{7}^{(\infty)}} + (i-7) - \left( \sum_{k=7}^{i-1} \frac{2}{k} + \frac{5 \ln k}{k^2} \right)  
\\
& >  \frac{1}{\mu_{7}^{(\infty)}} + (i-7) + 2 \ln(6.5) - 2 \ln(i-0.5) - \frac{5}{6}(1+\ln 6)
\\
& >  i - 2 \ln(i) +2,
\end{align}
where we have again used the inequality 
$\frac{1}{\left( x + \frac{1}{2} \right)} \leq \ln(x+1) - \ln(x)$, 
and also the bound $\sum_{k=7}^\infty\frac{\ln k}{k^2}<\int_6^{\infty} \frac{\ln x}{x^2} dx = (1+\ln 6)/6$.  
It remains to show that
\begin{align}
\frac{1}{  i -2\ln(i) + 2}
< \frac{1}{i} +  \frac{2\ln i}{i^2}. 
\end{align}
This is equivalent to $i+2\ln(i)>2\ln^2(i)$, which indeed holds for all $i\geq 1$ as one can check by taking derivatives.
\end{claimproof}

This completes the proof of the lemma.
\end{proof}

We can obtain an upper bound on the number of iterations $n$ that is 
sufficient to suppress any initial noise $\e \in (2/3,1)$ to 
less than $1/3$.  
\begin{lemma}\label{claim:n-iter-bound}
For any $\e \in (2/3,1)$, if an integer $n$ satisfies 
\begin{equation}
\label{eq:part3a}
n \geq \frac{1}{1-\e} \left(1 + 2 (1-\e) \ln \left( \frac{1}{1-\e} \right)  \right)
= \frac{1}{1-\e} + 2 \ln \left( \frac{1}{1-\e} \right),
\end{equation} 
then
$\mu_{n}^{(\infty)} \leq 1-\e$, and thus $\e_{n+1}^{(\infty)} < 2/3$.
\end{lemma}

\begin{proof}
Using \cref{claim:mu-inf-bound}, 
\begin{equation}
 \mu_n^{(\infty)} < \frac{1}{n} \left( 1 + \frac{2 \ln n}{n} \right) 
< \frac{1}{n} \left( 1 + 2 (1-\e) \ln \frac{1}{1-\e} \right) \leq 1-\e,
\end{equation}
where the second inequality holds because $n > \frac{1}{1-\e}$ 
and $\frac{\ln x}{x}$ 
is decreasing for $x\geq 3$.  
\end{proof}

As a comparison, if we use the rough upper bound 
for $\e_i^{(\infty)}$ in \cref{eq:older-delta-bound}, we conclude that the upper bound is less than 
$2/3$ after $n$ iterations if  
\begin{equation}
  n \geq \ceiling*{\frac{\ln (3\e/2) }{\ln (2-2\e+\e^2)}}.
  \label{eq:n'}
\end{equation}
When $1-\e \ll 1$, 
$\ln (2-2\e+\e^2) = \ln (1+(1-\e)^2) \approx (1-\e)^2$, so the upper bound on the number of iterations is roughly 
\begin{equation}
n \geq \ceiling*{\frac{\ln (3\e/2) }{(1-\e)^2}},
\end{equation}
with quadratically worse dependence on $\frac{1}{1-\e}$.  

Combining \cref{part:small,part:medium,part:large-a}, we have established the following estimate.
For all $d \geq 2$, $\e >
2/3$, and $\epsilon$ with $\e > \epsilon > 0$, the error can be reduced from $\e$ to $\epsilon$ with at most
\begin{equation}
    \log_2 \frac{1}{\epsilon} + 5 + \frac{1}{1-\e} + 2 \ln \left( \frac{1}{1-\e} \right)
\end{equation}
iterations.

\subsubsection{\Cref{part:large-b}}
\label{sec:part3b}

We now provide a more refined analysis\footnote{In the current analysis, \cref{claim2d} is the finite-$d$ analogue
of \cref{claim2} in \Cref{part:large-a}.  One can show that \Cref{claim2d} has the same
coefficients as in eq.~(74) of Part 3(b) in version 1 of the arXiv preprint of this manuscript.
Whereas the latter is an
asymptotic expression, \cref{claim2d} is a strict inequality.
Combining $\e_i^{(d)} \leq \e_i^{(\infty)}$ from \Cref{part:large-a}
and \cref{claim2d} gives a bound on the number of iterations
that subsumes the bounds in Parts 3(b) and (c) of version 1 of the
arXiv preprint, and asymptotically attains the conjectured expression
in eq.~(99) there.}
of $\e_i^{(d)}$ for finite $d$
and $\e \approx 1$.  The method is similar to \Cref{part:large-a}
above,
with more complex arithmetic and one more step to account for a
factor that is exponentially decaying in $i$.  
Wherever
appropriate, we omit straightforward calculations, and ideas
covered in \Cref{part:large-a} are described in less detail.

The goal is to upper bound the recurrence
\begin{align}
  \e_0^{(d)} &= \delta, \quad 
  \e_i^{(d)} = \D(\e_{i-1}^{(d)},d)
  = \frac{\e_{i-1}^{(d)} +  \frac{1}{d} \; \e_{i-1}^{(d)2}} { 2  \left(
  1 - \of[\Big]{1 - \frac{1}{d}} \e_{i-1}^{(d)} + \frac{1}{2} \of[\Big]{1 - \frac{1}{d}} \e_{i-1}^{(d)2} \right)}
\end{align}
stated in \cref{eq:recurrence,eq:PD}.
As in \cref{part:large-a}, let $i^*$ be the smallest
integer so that $\e_{i^*+1}^{(d)} < 2/3$.  
For $i \in \{0,1,\ldots, i^*\}$, let 
\begin{equation}
    \kappa_i^{(d)} = 1-\e_i^{(d)},
\end{equation}
which implies
\begin{equation}
\kappa_0^{(d)} 
= 1-\e, \quad
\kappa_i^{(d)} = 1 - \e_i^{(d)} 
= \frac{\left( 1 + \frac{2}{d}  \right) \kappa_{i-1}^{(d)} +  \left( 1 - \frac{2}{d}  \right)  \kappa_{i-1}^{(d)2}}{
        \left( 1 + \frac{1}{d}  \right) +  \left( 1 - \frac{1}{d}  \right)  \kappa_{i-1}^{(d)2}},
\label{eq:kappa}
\end{equation}
where we have first expressed $\e_i^{(d)}$ in terms of $\e_{i-1}^{(d)}$ and
then substituted $\e_{i-1}^{(d)} = 1-\kappa_{i-1}^{(d)}$ to obtain the recurrence.
Let
$g^{(d)}(x) = \frac{\left( 1 + \frac{2}{d} \right) x
+ \left( 1 - \frac{2}{d} \right) x^2}{ \left( 1 +
\frac{1}{d} \right) + \left( 1 - \frac{1}{d} \right) x^2}$ 
so that $\kappa_{i}^{(d)} = g^{(d)}(\kappa_{i-1}^{(d)})$.  
Taking the derivative of $g^{(d)}(x)$, we can show that
$g^{(d)}(x)$ is strictly increasing
if $1 + \left( \frac{d-2}{d+2} \right) 2x
- \left( \frac{d-1}{d+1} \right) x^2 > 0$. Let
$S$ be the set of $x$ satisfying this condition.
Note that $[0,1] \subset S$ for all $d \geq 2$.  
On the set $g^{(d)}(S)$, define an inverse $h^{(d)}$ of $g^{(d)}$.  
For $i \in \{0,1,\ldots, i^*\}$, 
define an inverse recurrence for $\kappa_i^{(d)}$ as 
\begin{equation}
\mu_0^{(d)} 
= \kappa_{i^*}^{(d)}, \quad
\mu_i^{(d)} = h^{(d)}(\mu_{i-1}^{(d)}) ,
\end{equation}
so indeed, $\mu_i^{(d)} = \kappa_{i^*-i}^{(d)}$.
In particular, $\mu_{i^*}^{(d)} = 1-\e$.
The definition of $i^*$ and the relations between 
$\e_i^{(d)}, \kappa_i^{(d)}, \mu_i^{(d)}$ for a representative
$d=20$ are summarized in 
\Cref{fig:kappa-mu}. 

We can now state the finite-$d$ analogue of \cref{claim2} in \Cref{part:large-a}.

\begin{claimd}{\ref*{claim2}$^{(d)}$} \label{claim2d}
For all $i \geq 1$, $d \geq 2$, 
\begin{align}
\frac{1}{\mu_i^{(d)}} & > \; \frac{1}{a} \; \frac{1}{\mu_{i-1}^{(d)}} + b - 2 \, c \, \mu_{i-1}^{(d)} ,
\label{eq:claim2d}
\end{align}
where \\[-5ex]
\begin{align}
a = \frac{ d+1 }{ d+2 }, \quad  
b = \frac{ d-2 }{ d+2 }, \quad  
c = \frac{d^3}{(d+2)^3} ~~\text{if}~d\geq 3,~~
c=\frac{1}{7} ~~\text{if}~d=2.
\end{align}
\end{claimd}
\vspace*{0.5ex}

\begin{proof}
Noting the positivity of the right-hand side of
\cref{eq:claim2d}, 
we can rephrase our claim as 
\begin{equation}
  \mu_i^{(d)} = h^{(d)}(\mu_{i-1}^{(d)})
  < \frac{1}{\frac{1}{a} \; \frac{1}{\mu_{i-1}^{(d)}} + b - 2 \, c \, \mu_{i-1}^{(d)}}  \,. 
\label{eq:claim2d2}
\end{equation}
Using $a \leq 1$, $\mu_{i-1}^{(d)} \leq 1/3$, $b \geq 0$, and $c \leq 1$, 
the right-hand side of \cref{eq:claim2d2} can be upper bounded by $3/7$, which is
within the range $S$ where $g^{(d)}$ is strictly increasing, and thus
invertible, so we can insert
the identity function $h^{(d)} g^{(d)}$ and apply the monotonicity of
$h^{(d)}$ (which follows from that of $g^{(d)}$) to obtain a statement
equivalent to \cref{eq:claim2d2}: 
\begin{equation}
  \mu_{i-1}^{(d)} < g^{(d)} \left(
  \frac{1}{\frac{1}{a} \; \frac{1}{\mu_{i-1}^{(d)}} + b - 2 \, c \, \mu_{i-1}^{(d)}}
  \right).
\label{eq:newcondclaim2d}
\end{equation}
As in \Cref{part:large-a}, the above steps translate the analysis from $h^{(d)}$
to $g^{(d)}$.
Furthermore, for $x>0$,
$g^{(d)}(x) = \frac{\left( 1 + \frac{2}{d} \right) x
+ \left( 1 - \frac{2}{d} \right) x^2}{ \left( 1 +
\frac{1}{d} \right) + \left( 1 - \frac{1}{d} \right) x^2} 
= \frac{\left( d+2 \right) \frac{1}{x} 
+ \left( d-2 \right) }{ \left( d+1 \right) \frac{1}{x^2}+ \left( d-1 \right)}$,
so \cref{eq:newcondclaim2d}
is equivalent to
\begin{align}
  \mu_{i-1}^{(d)} < 
  \frac{(d+2) \left( \frac{1}{a} \; \frac{1}{\mu_{i-1}^{(d)}} + b - 2 \, c \, \mu_{i-1}^{(d)} \right) + (d-2)}{(d+1) \left( \frac{1}{a} \; \frac{1}{\mu_{i-1}^{(d)}} + b - 2 \, c \, \mu_{i-1}^{(d)} \right)^2 + (d-1) } \,,
\end{align}
which in turn is equivalent to
\begin{align}
  0 ~ & < 
  (d+2) \left( \frac{1}{a} \; \frac{1}{\mu_{i-1}^{(d)}} + b - 2 \, c \, \mu_{i-1}^{(d)} \right)
\nonumber
\\[-3ex] 
    &\hspace*{23ex} + (d-2) - \mu_{i-1}^{(d)} \left( (d+1) \left( \frac{1}{a} \; \frac{1}{\mu_{i-1}^{(d)}} + b - 2 \, c \, \mu_{i-1}^{(d)} \right)^2 + (d-1)  \right)  
\nonumber
\\[1ex]
& = 2c (d+2)  \mu_{i-1}^{(d)} - (d-1) \mu_{i-1}^{(d)} - b^2 (d+1) \mu_{i-1}^{(d)}
- 4 c^2 (d+1) \mu_{i-1}^{(d) 3} + 4bc(d+1) \mu_{i-1}^{(d) 2}
\nonumber
\\[1ex]
& = \left( 2c (d+2) - (d-1) - b^2 (d+1) \right) \; \mu_{i-1}^{(d)} + 
\left(b-c \mu_{i-1}^{(d)} \right) \; 4 c (d+1) \mu_{i-1}^{(d)2},
\end{align}
where the second line is obtained from expanding the first
line and using the expressions for $a,b$ to cancel out
some of the terms.  
For $d \geq 3$, the choice of $c$ gives
$2 c (d+2) - (d-1) - b^2 (d+1) = 0$, and since $\mu_{i-1}^{(d)} < 1/3$, we have
$b-c \mu_{i-1}^{(d)} > 0$, so the claimed inequality holds.  
For $d=2$, $b=0$. One can substitute the 
actual values of $b,c,d$ and use $\mu_{i-1}^{(d)} < 1/3$ 
to verify the above inequality.
\end{proof}

Examining \cref{claim2d}, and observing that $\mu_{i-1}^{(d)} \leq
\mu_{i-1}^{(\infty)} \rightarrow 0$ as $i \rightarrow \infty$, the
term linear in $c$ has less effect on bounding $\mu_{i-1}^{(d)}$
for large $i$.  Thus we expect $\mu_{i}^{(d)}$ to
decrease at least as quickly as $a^i \mu_{0}^{(d)}$, with additional 
suppression due to the term linear in $b$. 
This motivates a change of variables
\begin{equation}
\mu_{i}^{(d)} = a^i \nu_{i}^{(d)},
\end{equation}
and correspondingly, \cref{claim2d} becomes 
\begin{align}
\frac{1}{\nu_i^{(d)}} & > \frac{1}{\nu_{i-1}^{(d)}} + a^i b - 2 \, a^i c \, \mu_{i-1}^{(d)} \,. 
\end{align}
Repeated application of the above gives 
\begin{align}
\frac{1}{\nu_i^{(d)}} & > \frac{1}{\nu_{0}^{(d)}} + b \sum_{k=1}^i a^k - 2 \, c \, \sum_{k=1}^i a^k \mu_{k-1}^{(d)} \,. 
\label{eq:claim2dnu}
\end{align}
Recall that a lower bound on $\frac{1}{\nu_i^{(d)}}$ gives an upper
bound on $\nu_i^{(d)}$, and in turn on $\mu_i^{(d)}$.
To this end, we take the following steps:
\begin{enumerate}
\item We sum the second term of the right-hand side of \cref{eq:claim2dnu}, and
apply the upper bound 
$\mu_{k}^{(d)} \leq \mu_{k}^{(\infty)} \leq
\frac{1}{k} + \frac{5}{2} \frac{\ln k}{k^2}$ 
for $k \geq 1$ from \Cref{part:large-a} and
$\mu_{0}^{(d)} \leq 1/3$ to the last term 
to obtain 
\begin{align}
\frac{1}{\nu_i^{(d)}} & > \; \frac{1}{\nu_{0}^{(d)}} + b \; a \; \frac{1-a^i}{1-a} 
- 2 \, c \, \sum_{k=1}^{i-1} \; a^{k+1}  
\left( \frac{1}{k} + \frac{5}{2} \frac{\ln k}{k^2} \right) - \frac{2}{3} \; c \; a .
\label{eq:claim2dnubdd}
\end{align}

\item For $0\leq a <1$, we can upper bound
$\sum_{k=1}^{i-1} \frac{a^k}{k}$ by $f(a) = \sum_{k=1}^{\infty} \frac{a^k}{k}$. 
Note that $\frac{df}{da} = \sum_{k=1}^{\infty} a^{k-1} = \frac{1}{1-a}$, so  
integrating and using $f(0) = 0$, we find $f(a) = -\ln(1-a) = \ln(d+2)
\lesssim \ln(d) + 0.7$.

\item We can also upper bound $\sum_{k=1}^{i-1} \frac{a^k}{k}$
by $\sum_{k=1}^{i-1} \frac{1}{k}
\leq 1 + \sum_{k=2}^{i-1} \ln(k+\frac{1}{2}) - \ln(k-\frac{1}{2})
= 1 +  \ln(i-\frac{1}{2}) - \ln(1.5) < \ln(i-\frac{1}{2}) + 0.6$, 
where we have used \cref{eq:lnbound}.   

\item\ First, 
$\sum_{k=1}^{i-1} \; a^{k-1} \frac{\ln k}{k^2}
< \sum_{k=2}^{i-1} \; \frac{\ln k}{k^2}
\lesssim 0.1733 + \sum_{k=3}^{i-1} \; \frac{\ln k}{k^2}$.  
Then, since $\frac{\ln x}{x^2}$ is decreasing for $x
\geq 2$, we have, for all $k \geq 3$, $\frac{\ln k}{k^2} <
\int_{k-1}^k \frac{\ln x}{x^2} dx$, so $\sum_{k=3}^{i-1} \; \frac{\ln
  k}{k^2} \leq \int_2^{\infty} \frac{\ln x}{x^2} dx$.  Finally, 
$\int \frac{\ln x}{x^2} dx = -\frac{1+\ln x}{x} +$constant, so
$\int_2^{\infty} \frac{\ln x}{x^2} dx \lesssim 0.8466$.  Altogether, we 
have $\sum_{k=1}^{i-1} \; a^{k-1} \frac{\ln k}{k^2} < 1.02$.

\item Substituting the bounds from steps 2--4 into \cref{eq:claim2dnubdd}, 
\begin{align}
\frac{1}{\nu_i^{(d)}} & > \; \frac{1}{\nu_{0}^{(d)}} + b \; a \; \frac{1-a^i}{1-a} 
- 2 \, c \, a \; \ln \left( \min\{d,i\} \right) - 1.4 \, c \, a \, - 5.1 \, c \, a^2
- \frac{2}{3} \; c \; a 
\\ & > \; b \; a \; \frac{1-a^i}{1-a} 
- 2 \, c \, a \; \ln \left( \min\{d,i\} \right) + 3 - 7.2 \, c \, a.
\label{eq:claim2dtmp}
\end{align}

Note that \cref{eq:claim2dtmp} holds for all $i \geq 1$.  The cases $i=1,2$ can be verified directly.  Steps 2--4 hold for all $i
\geq 3$, except for the expression $\sum_{k=3}^{i-1} \; \frac{\ln
  k}{k^2}$ in step 4, which can be set to $0$ for $i=3$ without
changing the conclusion.
Reverting the change of variables $\mu_{i}^{(d)} = a^i \nu_{i}^{(d)}$, we have
\begin{align}
\frac{1}{\mu_i^{(d)}}
  & > a^{-i} \left( \; b \; a \; \frac{1-a^i}{1-a} 
- 2 \, c \, a \; \ln \left( \min\{d,i\} \right) + 3 - 7.2 \, c \, a \right)
\\
& > 
a^{-i} \; \left( \frac{ba}{1-a} - 2 \, c \, a \; \ln \left( \min\{d,i\} \right)
+ 3 - 7.2 \, c \, a \right)
- \frac{ba}{1-a} .
\label{eq:muidbdd}
\end{align}
\end{enumerate}

Recall that $\mu_{i^*}^{(d)} = 1-\e$, and our ultimate goal is an
upper bound on $i^*$. Also recall that $i^*+1$ iterations reduce 
$\e_0^{(d)} = \e$ to
$\e_{i^*+1}^{(d)} < \frac{2}{3}$.  Furthermore, $i^* \leq n^{(\infty)*} \; {:}{=} \;
\frac{1}{1-\e} + 2 \ln \left( \frac{1}{1-\e} \right)$ (the upper bound
for $i^*$ for the $d\rightarrow \infty$ case in \cref{claim:n-iter-bound}). 
Incorporating all these observations into \cref{eq:muidbdd},
we obtain
\begin{align} \frac{1}{1-\e} = 
\frac{1}{\mu_{i^*}^{(d)}}
> 
a^{-{i^*}} \; \left( \frac{ba}{1-a} - 2 \, c \, a \; \ln \left( \min\{d, n^{(\infty)*}\} \right)
+ 3 - 7.2 \, c \, a \right)
- \frac{ba}{1-a} .
\label{eq:part3b-bd1}
\end{align}
This allows us to make a transition from the analysis of $\mu_{i}^{(d)}$ to $i^*$.
Let
\begin{equation}
\alpha = \frac{ba}{1-a} = \frac{(d-2)(d+1)}{d+2}, \quad
\beta = \alpha - 2 \, c \, a \; \ln \left( \min\{d, n^{(\infty)*}\} \right)
+ 3 - 7.2 \, c \, a.
\end{equation}
To be concrete, for $d=2,3,4,5$, $(\alpha, \beta) \approx (0,2.08),
(0.8,2.18), (1.67,2.20), (2.57, 2.32)$, respectively,  and as $d$ increases, $\alpha
\approx d$ and $\beta \approx d - 2 \ln(\min\{d,n^{(\infty)*}\})$.
\Cref{eq:part3b-bd1} can now be rephrased as
\begin{align} \frac{1}{1-\e}  
> 
a^{-{i^*}} \; \beta
- \alpha  , 
\end{align}
or
\begin{align}  
  (a^{-1})^{i^*} < 
  \frac{\frac{1}{1-\e} + \alpha}{\beta} .
\label{eq:part3b-bd2}
\end{align}
This gives an upper bound on $i^*$.

\begin{lemma}\label{claim:n-iter-bound-finite-d}
For any $\e \in (2/3,1)$, if an integer $n$ satisfies 
\begin{equation}
\label{eq:part3b-bd4}
n > \frac{\ln \left( \frac{\frac{1}{1-\e} + \alpha}{\beta} \right)}{\ln \frac{1}{a}}
= \begin{cases} \frac{\ln \frac{1}{(1-\e) \beta}}{\ln \frac{4}{3}} \approx  
3.476 \, \ln \frac{1}{1-\e} - 2.546 & \text{if $d=2$} 
\\[2ex] 
\frac{\ln \left( 1+ \frac{1}{\alpha (1-\e)} \right) + \ln \frac{\alpha}{\beta}}{\ln \left( 1+ \frac{1}{d+1} \right)} & \text{if $d\geq 3$},  \end{cases}  
\end{equation} 
then
$\e_{n}^{(d)} < 2/3$.
\end{lemma}

We can further understand the above lemma by making some approximations.
Except for very small dimension $d$, we have $a,b,c \approx 1$,
$\alpha \approx d$, $\beta \approx d - 2 \ln \min\left\{
d,\frac{1}{1-\e} \right\}$, and $\frac{\alpha}{\beta} \approx
1+\frac{2}{d} \ln \min\left\{ d,\frac{1}{1-\e} \right\}$.
For a small constant dimension $d$ and vanishing $1-\e$, we have $\frac{1}{\alpha (1-\e)} \gg 1$. Therefore, 
$1+ \frac{1}{\alpha (1-\e)} \approx \frac{1}{\alpha (1-\e)}$, 
so the number of iterations
is approximately upper bounded by $\frac{\ln \frac{1}{\beta (1-\e)}}{\ln
  1+\frac{1}{d}}$.  Furthermore, replacing $\beta$ by $d$ does not 
  change the leading-order behavior.
If instead, the dimension $d$ is large and $1-\e$ is small but
fixed, so $\alpha (1-\e) \gg 1$, 
then $\ln \left( 1+ \frac{1}{\alpha (1-\e)} \right) \approx 
\frac{1}{\alpha (1-\e)}$ while $\ln \left( 1+ \frac{1}{d+1} \right)
\approx  \frac{1}{d+1}$, and $\ln \frac{\alpha}{\beta} \approx
\frac{2}{d} \ln \frac{1}{1-\e}$.  
\Cref{eq:part3b-bd4} gives an upper bound for the number of iterations that is
approximately 
$(d+1) \left( \frac{1}{\alpha (1-\e)} + \frac{2}{d} \ln \frac{1}{1-\e} \right)
\approx \frac{1}{1-\e} +  2 \ln \frac{1}{1-\e} $, which is similar
to the bound in \cref{part:large-a}.

\subsection{Expected sample complexity}\label{sec:complexity}

Since the $\SWAP$ gadget in \cref{alg:swap} is intrinsically probabilistic,
the number of states it consumes in any particular run is a random variable.
The same applies to our recursive purification procedure $\PURIFY(n)$ (\cref{alg:purify}).
Hence we are interested in the \emph{expected sample complexity} of $\PURIFY(n)$, denoted by $\SC(n,d)$,
which we define as the expected number of states $\rho_0 = \rho(\e)$ consumed during the algorithm,
where the expectation is over the random measurement outcomes in all the swap tests used by
the algorithm, and $d$ is the dimension of $\rho(\e)$.
The sample complexity of $\PURIFY(n)$ is characterized by the following theorem.

\begin{theorem}\label{thm:swap}
Fix $d \geq 2$ and $\e \in (0,1)$. 
Let $\ket{\psi} \in \C^d$ be an unknown state and define $\rho(\e)$ as in \cref{eq:re}.
For any $\epsilon > 0$, there exists a procedure that outputs $\rho(\e')$ with $\e' \leq \epsilon$
by consuming an expected number of 
  $\SC(n,d) <
  4^{
       \min \left\{
                    \frac{1}{1-\e} + 2 \ln \left( \frac{1}{1-\e} \right) , 
                    (d+2) \ln  \left(  \frac{1}{1-\e} \right) 
       \right\}
  }
  \times \frac{3630}{\epsilon}$
copies of $\rho(\e)$. 
\end{theorem}

Our strategy for proving \cref{thm:swap} is as follows.  In
\Cref{claim:cost}, we derive an expression for the expected sample
complexity $\SC(n,d)$ of $\PURIFY(n)$ in terms of the recursion depth
$n$ and the product $\prod_{i=1}^n p_i^{(d)}$ of success probabilities
$p_i^{(d)}$ of all steps. Then we prove the theorem by lower
bounding the product $\prod_{i=1}^n p_i^{(d)}$ using \cref{eq:PD},
\Cref{claim:bounds}, and upper bounds on the number of iterations $n$
from \Cref{sec:asymptotics}.

\begin{lemma}\label{claim:cost}
The expected sample complexity of $\PURIFY(n)$ is
\begin{equation}
  \SC(n,d) = \frac{2^n}{\prod_{i=1}^n p_i^{(d)}},
  \label{eq:Cnd}
\end{equation}
where the probabilities $p_i^{(d)}$ are subject to the recurrence in \cref{eq:recurrence}.
\end{lemma}

\begin{proof}
For any $i \geq 1$, the expected number of copies of $\rho_{i-1}$ consumed by the $\SWAP$ gadget within $\PURIFY(i)$ is $2/p_i^{(d)}$, as in \cref{eq:exp_cpy}.
Since each $\rho_{i-1}$ in turn requires an expected number of $2/p_{i-1}^{(d)}$ copies of $\rho_{i-2}$, the overall number of states $\rho_0$ consumed is equal to the product of $2/p_i^{(d)}$ across all $n$ levels of the recursion.
The expectation is multiplicative because the measurement outcomes at different levels are independent random variables.
\end{proof}

\begin{proof}[Proof of \cref{thm:swap}]
To lower bound the denominator of $\SC(n,d)$ in
\cref{eq:Cnd}, 
we first lower bound the values $p_i^{(d)}$ by bounding \cref{eq:PD} as
\begin{align}
  P(\e, d) = 1 - \of[\Big]{1 - \frac{1}{d}} \e + \frac{1}{2} \of[\Big]{1 - \frac{1}{d}} \e^2
  > \max \left( 1-\e, \frac{1}{2} \of[\Big]{1 + \frac{1}{d}} \right)
  > \max \left(1-\e, \frac{1}{2}\right)
\label{eq:prob-bdd}.
\end{align}
The first bound of $1-\e$ is immediate, and the second bound 
is $P(1, d)$, 
which follows from the monotonicity of $P(\e, d)$ established in \Cref{claim:monotonicity}. 

We now combine \cref{claim:cost} and \cref{eq:prob-bdd} to determine the
number of iterations of \cref{alg:purify} to reach the
desired target error parameter $\epsilon > 0$, starting from an
arbitrary initial error parameter $\e \in (0,1)$, for an arbitrary
$d \geq 2$.  Following the structure of \cref{sec:asymptotics}, we
divide the proof into three cases, depending on the initial error
parameter $\e$.

\textbf{Case 1: $0 < \e < 1/2$}.
Using the final conclusion from \cref{part:small}, 
the number of iterations $n$ to achieve a final error 
parameter $\epsilon$ can be chosen to be no more than 
$\lceil \log_2 \frac{\e}{(1-2\e) \epsilon} \rceil \leq 
\log_2 \frac{\e}{(1-2\e) \epsilon} + 1$.  
We can thus upper bound the numerator of $\SC(n,d)$ as follows:
\begin{equation}
  2^n \leq 2^{1 + \log_2 \frac{\e}{(1-2\e) \epsilon}} 
  = 2 \cdot \frac{\e}{\epsilon} \cdot \frac{1}{1-2\e}.
  \label{eq:num}
\end{equation}
Using \cref{eq:012} in \cref{claim:bounds}, 
\begin{align}
  \prod_{i=1}^n p_i^{(d)}
  > \prod_{i=1}^n \left( 1 - \e_{i-1}^{(d)} \right)
  &\geq \prod_{i=1}^n \of*{1 - \frac{\e}{2^{i-1}(1-2\e)+2\e}} \label{eq:pid} \\
  &= 1 - 2 \e (1-2^{-n}) \label{eq:induction} \\
  &> 1 - 2 \e, \label{eq:den}
\end{align}
where \cref{eq:induction} can be proved by induction on $n$.
Combining \cref{eq:Cnd,eq:num,eq:den},
\begin{equation}
  \SC(n,d) \leq \frac{2\e}{\epsilon(1-2\e)^2}.
  \label{eq:SCn}
\end{equation}

\vspace*{2ex}

\textbf{Case 2: $1/3 \leq \e < 2/3$}.  From \cref{eq:prob-bdd} 
and the monotonicity of $P(\e, d)$ in \Cref{claim:monotonicity}, 
$p_i^{(d)} > P(2/3,\infty) = 5/9$.  Combining this with the factor 
of $2$ in the
numerator, each iteration uses at most $3.6$ noisy states in
expectation.
At most $5$ iterations are sufficient to reduce the
error from $2/3$ to less than $1/3$ (using the $d=\infty$ case as
an upper bound for all $d$), so $3.6^5 \leq 605$ samples suffice.  
Applying the analysis of the previous case for $\e=1/3$, we see that
$6/\epsilon$ samples reduce the
error from $1/3$ to at most $\epsilon$.  
Combining the bounds from these two stages, we have 
\begin{equation}
  \SC(n,d) \leq \frac{3630}{\epsilon} .
\end{equation}

\textbf{Case 3: $2/3 \leq \e < 1$}. 
Our strategy is to first drive the error parameter below $2/3$
and then apply the bound from case 2.  Recall from 
the previous subsection that this takes $i^*+1$ iterations, 
but one iteration is already counted in case 2.  
Since $p_i^{(d)} > \frac{1}{2} \left(1+\frac{1}{d}\right)$ 
from \cref{eq:prob-bdd}, 
$\SC(n,i^*) \leq \left(\frac{4}{1+\frac{1}{d}} \right)^{i^*}$.
For large enough $d$, the simpler complexity bound of $4^{i^*}$ suffices.  
\Cref{part:large-a,part:large-b} each provide a complexity
bound $\SC(i^*,d)$ for reducing the error from $\e$ to $2/3$, and
we can take the minimum of the two. 
From \Cref{claim:n-iter-bound} in \cref{part:large-a},
\begin{equation}
    \label{eq:bound_no_d}
    \SC(i^*,d) < 4^{i^*} \leq 4^{\frac{1}{1-\e} + 2 \ln \left( \frac{1}{1-\e} \right)} .
\end{equation}
For small dimension $d \geq 3$, \Cref{part:large-b} provides a sharper
bound of $d (1-\e) \ll 1$.  
In this case, we can use \cref{eq:part3b-bd2} to obtain 
\begin{align}  
   \SC(i^*,d) < 4^{i^*} = (a^{-1})^{\frac{\ln 4}{\ln \frac{d+2}{d+1}} i^*} <    
   \left(\frac{\frac{1}{1-\e} + \alpha}{\beta} \right)^{\frac{\ln 4}{\ln \frac{d+2}{d+1}}} 
   < 4^{\frac{ \ln  \left( {\frac{1}{1-\e} + \frac{\alpha}{\beta}} \right)}{\ln \frac{d+2}{d+1}}},
\label{eq:part3b-bd3}
\end{align}
where $\alpha,\beta$ are as defined in \Cref{part:large-b}.  
Combining the above bounds and that from the previous case, we upper bound
the sample complexity to reduce the noise parameter from $\e > 2/3$ to $\epsilon$: 
\vspace*{-3ex}
\begin{align}  
   \SC(n,d) < 4^{\min \left\{ \frac{1}{1-\e} + 2 \ln \left( \frac{1}{1-\e} \right) , 
     \frac{\ln \left( \frac{\frac{1}{1-\e} + \alpha}{\beta}\right)}{\ln \frac{d+2}{d+1}}   \right\}}  \times
   \frac{3630}{\epsilon} .
\end{align}
Given $d$, $\e$, the above provides a tight bound.  However, we can
also simplify the second bound by observing 
\begin{align}  
  \ln \left(  \frac{ \frac{1}{1-\e} + \alpha }{\beta} \right)
= \ln  \left(  \frac{1}{1-\e}\right) + \ln \left(   \frac{\alpha(1-\e)}{\beta} + \frac{1}{\beta} \right) .
\label{eq:small-d-bdd}
\end{align}
For large $d$, $\alpha \approx \beta \approx d$, so
$\frac{\alpha(1-\e)}{\beta} + \frac{1}{\beta} \rightarrow 1-\e < 1/3$.  
In fact, for all $d \geq 2$, 
\begin{align}
\frac{\alpha(1-\e)}{\beta} + \frac{1}{\beta} 
& < \alpha/3 + \frac{1}{\beta} \\
& < \frac{\left(\beta + 2 \, c \, a \; \ln ( \min\{d, n^{(\infty)*}\} )
- 3 + 7.2 \, c \, a\right)/3 +1}{\beta} \\
& < \frac{1}{3} + \frac{2 \, c \, a \; \ln d 
+ 7.2 \, c \, a}{3\beta} .
\end{align}
As $d$ increases, $ca \rightarrow 1$ and $\beta$ increases 
roughly as $d$, and the second term above initially increases to
attain a maximum of $0.5512$ at $d=7$ and decreases thereafter, 
so the last line is upper bounded by $0.8845$. Applying 
this bound to \cref{eq:small-d-bdd} gives
\begin{align}  
  \ln \left(  \frac{\frac{1}{1-\e} + \alpha }{\beta} \right)
< \ln  \left(  \frac{1}{1-\e}\right) .
\end{align}
Using
$\ln(1+x) \geq \frac{x}{x+1}$ to get
$\ln \frac{d+2}{d+1} \geq \frac{1}{d+2}$,
we have
\begin{align}  
  \SC(n,d) <
  4^{
       \min \left\{
                    \frac{1}{1-\e} + 2 \ln \left( \frac{1}{1-\e} \right) , 
                    (d+2) \ln  \left(  \frac{1}{1-\e} \right) 
       \right\}
  }
  \times \frac{3630}{\epsilon}
\end{align}
as claimed.
\end{proof}

As a side note, 
if $\alpha (1-\e) \ll 1$, we can obtain a tighter approximation (though not a strict bound)
for the sample complexity bound by revising \cref{eq:small-d-bdd} as
\begin{align}  
  \ln \left( \frac{\frac{1}{1-\e} +  \alpha }{\beta}  \right)
= \ln  \left(  \frac{1}{1-\e}\right) + \ln \left(   \frac{\alpha(1-\e)}{\beta} + \frac{1}{\beta}  \right) 
\approx \ln  \left(  \frac{1}{1-\e}\right) + \ln \left( \frac{1}{\beta} \right) 
=  \ln  \left(  \frac{1}{(1-\e) \beta}\right) ,
\nonumber
\end{align}
\vspace*{1ex}
$\!\!$so 
\vspace*{-3ex}
\begin{align}  
  \SC(n,d) \lesssim
  4^{ \frac{ \ln  \frac{1}{(1-\e)\beta}}{ \ln \frac{ d+2 }{ d+1} }} 
  \times \frac{3630}{\epsilon} , 
\end{align}
which recovers a polynomial dependence on $\frac{1}{1-\e}$.

We conclude this section by describing the gate complexity of $\PURIFY$. Since the number of internal nodes in a binary tree is upper bounded by the number of leaves,
the total number of $\SWAP$ gadgets is upper bounded by twice the sample complexity. As mentioned in \cref{sec:swaptest}, the swap test for qudits can be performed using $O(\log d)$ two-qubit gates. Therefore, from \cref{eq:bound_no_d} we have the following upper bound on the gate complexity of purification for constant $\delta$.

\begin{corollary}
For any dimension $d \geq 2$ and any fixed initial error parameter $\e \in (0,1)$, the expected gate complexity of $\PURIFY$
to achieve an error parameter $\epsilon$ is $O(\frac{1}{\epsilon} \log d)$.
\label{cor:gate_comp}
\end{corollary}

\section{Applications} \label{sec:apps} 

\subsection{Simon's problem with a faulty oracle}\label{sec:simon} 

In this subsection we describe an application of quantum state purification to a query complexity problem with an inherently faulty oracle. Query complexity provides a model of computation in which a black box (or oracle) must be queried to learn information about some input, and the goal is to compute some function of that input using as few queries as possible. It is well known that quantum computers can solve certain problems using dramatically fewer queries than any classical algorithm.

Regev and Schiff \cite{RS08} studied the effect on query complexity of providing an imperfect oracle. They considered the unstructured search problem, where the goal is to determine whether some black-box function $f\colon[M] \to \{0,1\}$ has any input $x \in [M]$ for which $f(x)=1$. This problem can be solved in $O(\sqrt M)$ quantum queries \cite{Grover96}, whereas a classical computer needs $\Omega(M)$ queries. Regev and Schiff showed that if the black box for $f$ fails to act with some constant probability, the quantum speedup disappears: then a quantum computer also requires $\Omega(M)$ queries. Given this result, it is natural to ask whether a significant quantum speedup is ever possible with a faulty oracle.

Here we consider the quantum query complexity of Simon's problem \cite{S97}. In this problem we are given a black-box function $f \colon \{0,1\}^m \to \{0,1\}^m$ with the promise that there exists a hidden string $s \in \{0,1\}^m \setminus \{0^m\}$ such that $f(x) = f(y)$ if and only if $x = y$ or $x \oplus y = s$. The goal is to find $s$. With an ideal black box, this problem can be solved with $O(m)$ quantum queries (and $\poly(m)$ additional quantum gates), whereas it requires $2^{\Omega(m)}$ classical queries (even for a randomized algorithm).

Now we consider a faulty oracle. If we use the same model as in \cite{RS08}, then it is straightforward to see that the quantum query complexity remains $O(m)$. Simon's algorithm employs a subroutine that outputs uniformly random strings $y$ with the property that $y \cdot s=0$. By sampling $O(m)$ such values $y$, one can determine $s$ with high probability. If the oracle fails to act, then Simon's algorithm is instead guaranteed to output $y=0$. While this is uninformative, it is consistent with the condition $y \cdot s=0$. Thus it suffices to use the same reconstruction procedure to determine $s$, except that one must take more samples to get sufficiently many nonzero values $y$. But if the probability of the oracle failing to act is a constant, then the extra query overhead is only constant, and the query complexity remains $O(m)$.

Of course, this analysis is specific to the faulty oracle model where the oracle fails to act with some probability. We now consider instead another natural model of a faulty oracle where, with some fixed probability $\e \in (0,1)$, the oracle depolarizes its input state. In other words, the oracle acts on its input density matrix $\rho$ according to the quantum channel $D_\e$ where
\begin{equation}
  D_\e(\rho) = (1-\e) U_f \rho U_f^\dagger+ \e \frac{I}{2^{2m}}
\end{equation}
where $U_f$ is the unitary operation that reversibly computes $f$ (concretely, $U_f\ket{x,a}=\ket{x,a \oplus f(x)}$).
In this case, if we simply apply the main subroutine of Simon's algorithm using the faulty oracle, we obtain a state
\begin{equation}
  \rho' = (1-\e)\ket{\Psi}\bra{\Psi} + \e \frac{I}{2^{2m}}
\end{equation}
before measuring the first register,
where
\begin{equation}
  \ket{\Psi} = \frac{1}{2^m} \sum_{x \in \{0,1\}^m} \sum_{\substack{y \in \{0,1\}^m \\ \text{s.t. } y \cdot s = 0}} (-1)^{x \cdot y}\ket{y} \ket{f(x)}
\end{equation}
is the ideal output.

If we simply measure the first register of $\rho'$, then with probability $1-\e$ we obtain $y$ satisfying $y \cdot s = 0$, but with probability $\e$, we obtain a uniformly random string $y$. Determining $s$ from such samples is an apparently difficult ``learning with errors'' problem. While $s$ is information-theoretically determined from only $O(m)$ samples---showing that $O(m)$ quantum queries suffice---the best known (classical or quantum) algorithm for performing this reconstruction takes exponential time \cite{R10}.

We can circumvent this difficulty using quantum state purification. Suppose we apply $\PURIFY$ to copies of $\rho'$, each of which can be produced with one (faulty) query. Since Simon's algorithm uses $O(m)$ queries in the noiseless case, purifying the state to one with $\epsilon = O(1/m)$ ensures that the overall error of Simon's algorithm with a faulty oracle is at most constant. By \cref{thm:swap}, the expected number of copies of $\rho'$
consumed by $\PURIFY$ is $O(\frac{1}{\epsilon})= O(m)$. We repeat this process $O(m)$ times to obtain sufficiently many samples to determine $s$.
Overall, this gives an algorithm with query complexity $O(m^2)$. Since the purification can be performed efficiently by \cref{cor:gate_comp}, this algorithm uses only $\poly(m)$ gates.

Note that the same approach can be applied to any quantum query algorithm that applies classical processing to a quantum subroutine that makes a constant number of queries to a depolarized oracle. However, this approach is not applicable to a problem such as unstructured search for which quantum speedup requires high query depth \cite{Zal99}.

\subsection{Relationship to mixedness testing and quantum state tomography}\label{sec:mt-qst} 

In the problem of \emph{mixedness testing} \cite{MdW13}, copies of an
unknown $d$-dimensional quantum state $\rho$ are available, where
$\rho$ is promised to be either the maximally mixed state
$\frac{I}{d}$ or at least $\eta$-far from $\frac{I}{d}$ in trace
distance.
References \cite{DW15,Wright-thesis-16} establish a sample complexity
of $\Theta(d/\eta^2)$ copies of $\rho$ for this problem. 
In particular, the
lower bound is proved for a state $\rho$ with eigenvalues
$\frac{1}{d}(1 \pm \frac{\eta}{2})$, each with degeneracy
$\approx \frac{d}{2}$.

If the state $\rho$ is restricted to the form $\rho(\e)$ given by
\cref{eq:re}, where $\e=1$ or $\e \leq 1-\frac{\eta}{2}$, one can
solve the mixedness testing problem simply by applying \cref{alg:purify}
as an attempt to purify the samples.  If $\e=1$, the swap
test in each iteration fails with probability $1/2$ and 
always outputs $I/d$. If instead, $\e \leq 1-\frac{\eta}{2}$, 
then after
$15 + \frac{1}{1-\e} + 2 \ln \left( \frac{1}{1-\e} \right)
\leq 15 + \frac{2}{\eta} + 2 \ln  \frac{2}{\eta}$ 
iterations, the algorithm outputs $\rho(\epsilon)$ with
$\epsilon \leq 2^{-10}$, and the swap test in the final
iteration passes with near certainty.
The difference in the probability of passing the swap test in the
last iteration can be used to determine whether 
$\e=1$ or $\e \leq 1-\frac{\eta}{2}$ with only a few
repetitions.  
The sample complexity found in \Cref{sec:complexity} has worse
dependence on $\eta$ (compared to $\Theta(d/\eta^2)$ from
\cite{DW15,Wright-thesis-16}), but it is independent of the dimension $d$.
Therefore, under the promise that the unknown $\rho$ has the form of a depolarized pure state as in \cref{eq:re}, 
the protocol in \cite{DW15,Wright-thesis-16} is
not optimal in terms of the dimension $d$.

In the problem of quantum state tomography, copies of an
unknown $d$-dimensional quantum state $\rho$ are available, and a
desired accuracy level $\eta>0$ is specified.  The output is a
classical description of a $d$-dimensional quantum state
$\sigma$ such that $\|\rho-\sigma\|_1 \leq \eta$.  
For this application, we adopt the 
trace distance as a measure of accuracy which better facilitates the discussion.

One method to perform quantum state purification of an unknown 
state $\rho$ of the form $\rho(\e) = (1 - \e) \proj{\psi}
+ \e \, \frac{I}{d}$ given by \cref{eq:re} is to
first perform quantum state tomography, obtain an estimate $\sigma$ of
$\rho(\e)$, and then prepare a
copy of the principal eigenvector of $\sigma$.
Let $\ket{\phi}$ be the principal eigenvector of $\sigma$.
We now determine how many samples suffice to achieve $\|\proj{\psi} - \proj{\phi} \|_1 \leq \epsilon$.

Since $\proj{\psi}$ is scaled by a factor 
of $1-\e$ in $\rho(\e)$, we expect the required accuracy for tomography $\eta$ to 
scale linearly in $1-\e$.  Therefore, we let $\eta = (1-\e) \xi$ and 
determine $\xi$ as a function of the allowed noise parameter
$\epsilon$ for the purification.  Thus $\xi$ relates the accuracy of tomography to that of purification.

The accuracy of this approach to purification can be understood as follows.  
If $\| \rho(\e) - \sigma \|_1 \leq \eta$
and $\lambda$ is the principal eigenvalue of $\sigma$, 
then
\begin{align}
\lambda \geq (1-\e) + \e/d - (1-\e) \xi = 1/d + (1-\e)(1-\xi - 1/d).
\end{align}
Furthermore, since $\ket{\phi}$ is the principal eigenvector
of $\sigma$, 
\begin{align}
(1-\delta) \xi \geq \Tr \proj{\phi} (\sigma-\rho(\e)) =   
|\lambda - (1-\e) | \langle \phi | \psi \rangle |^2 - \e/d|,
\end{align}
where we have used 
$\| X \|_1 = \max\{ \Tr(UX) : \text{$U$ unitary}\}$ for a square matrix $X$ 
to obtain the leftmost inequality. 
Combining the two inequalities above, 
$| \langle \phi | \psi \rangle |^2 \geq 1-2\xi$.  The trace
distance 
between $\ket{\psi}$ and $\ket{\phi}$ is
$\sqrt{1-| \langle \phi | \psi \rangle |^2}
\leq \sqrt{2\xi} = \epsilon$ if we choose $\xi = \epsilon^2/2$.  

Having related $\xi$ to $\epsilon$, it remains to 
determine the sample complexity for the tomography step.
The sample complexity is
$O(d^2/\eta^2) = O\left(d^2 (1-\e)^{-2} \epsilon^{-4} \right)$ 
if collective measurements
are allowed, and $O(d^3/ \eta^2) 
= O(d^3 (1-\e)^{-2} \epsilon^{-4})$ if
one is restricted to single-copy (streaming) measurements  
(see for example \cite{OW15-tomo,Wright-thesis-16,HHJWY15}).  We remark that for $d=2^n$,
these measurements can be chosen to be single-qubit Pauli
measurements, which can be efficiently performed.  
Additional structure of the states $\rho(\e)$ given by \cref{eq:re},
such as low rank of the pure-state component and tracelessness
of the noisy component, likely reduce the complexity further.  We
leave a detailed analysis to future research.

Comparing the sample complexity for purification using
\cref{alg:purify} (see \Cref{sec:complexity}) with that using state
tomography, when $d$ is small 
and $1-\e$ is vanishing, state
tomography has better dependence on $1-\e$ (even with single-copy
measurements).
One disadvantage of \cref{alg:purify} is that many copies
of $\rho(\e)$ are discarded in the process when the 
swap test fails, until $\e$
drops below $1/2$.  This is a major source of inefficiency 
as the postmeasurement quantum states from 
failed swap tests still contain significant  
information about $\rho(\e)$. 
In contrast, state tomography does not waste any samples,
and the amplification of the small signal has only 
inverse square dependence on $1-\e$.  
However, for any constant $\e$, or for growing dimension $d$
and small but non-vanishing $1-\e$, \cref{alg:purify} has
dimension-independent complexity, unlike state
tomography (which involves learning the state and
necessarily has complexity growing with dimension).  
One additional disadvantage of tomography is 
that it can output a target state that has high 
quantum circuit complexity to prepare.

\subsection{Relationship to quantum majority vote problem} \label{sec:q-majority} 

Purification is closely related to the following \emph{quantum majority vote} problem \cite{Majority}: given an $N$-qubit state $U\xp{N} \ket{x}$ for an unknown unitary $U \in \U{2}$ and bit string $x \in \set{0,1}^N$, output an approximation of $U \, \ket{\mathrm{MAJ}(x)}$, where $\mathrm{MAJ}(x) \in \set{0,1}$ denotes the majority of $x$.
While the input to this problem is not of the form $\rho\xp{N}$, applying a random permutation produces a separable state whose single-qubit marginals are identical and equal to a state $\rho$ as in \cref{eq:re}.
The optimal quantum majority vote procedure for qubits \cite{Majority} appears to coincide with the optimal qubit purification procedure \cite{CEM99}.
However, its generalization to qudits has been investigated only in very limited cases \cite{LP}.
Our qudit purification problem can be interpreted as a generalization of quantum majority voting to higher dimensions.

\section{Sample complexity lower bound}\label{sec:lowbdd} 

In this section we prove a lower bound on the cost of purifying a $d$-dimensional state. In particular, this bound shows that for constant $d$, the $\epsilon$-dependence of our protocol is optimal.

\begin{lemma}
  \label{lm:embed}
  For any constant input error parameter $\e>0$,
  the sample complexity of purifying a $d$-dimensional state $\rho({\e})$ with output fidelity $1-\epsilon$ is $\Omega(\frac{1}{d\epsilon})$.
\end{lemma}

\begin{proof}
We use the fact that a purification procedure for $d$-dimensional states can also be used to purify $2$-dimensional states, a task with known optimality bounds.  
For clarity, we write $\ket{1}$ as the target qubit state to be purified, and let $\ket{2}$ denote a state orthonormal to $\ket{1}$.  This is simply a choice of notation and entails no loss of generality, as the operations described in the proof do not depend on the states $\ket{1}$ or $\ket{2}$.  

We consider two procedures for purifying $N$ samples of the qubit state $\rho^{(2)}(\e^{(2)}) = (1-\frac{\e^{(2)}}{2})\ket{1}\bra{1}+\frac{\e^{(2)}}{2}\ket{2}\bra{2}$:
 \begin{enumerate}
 \item applying the optimal qubit procedure, denoted by $\mathcal{P}_{\mathrm{opt}}^{(2)}$, on the $N$ samples; or
  \item embedding each copy of $\rho^{(2)}(\e^{(2)})$ in the $d$-dimensional space and then applying the optimal qudit purification procedure, denoted by $\mathcal{P}_{\mathrm{opt}}^{(d)}$.
\end{enumerate}
We denote the second procedure by $\mathcal{P}_{\mathrm{embed}}^{(2)}$. Clearly it cannot outperform the optimal procedure $\mathcal{P}_{\mathrm{opt}}^{(2)}$.

The embedding operation involves 
mixing the given state $\rho^{(2)}(\e^{(2)})$ with the maximally mixed state $\frac{\Lambda}{d-2}$ in the space orthogonal to the support of $\rho^{(2)}(\e^{(2)})$.  
The embedded $d$-dimensional state is
\begin{equation}
  \label{eq:rhod_form1}
  \rho^{(d)}(\e^{(d)}) = (1-q)\rho^{(2)}(\e^{(2)}) + \frac{q}{d-2} \Lambda
\end{equation}
where
$\Lambda := \sum_{i=3}^{d} \ket{i}\bra{i}$.
We want $\rho^{(d)}(\e^{(d)})$ to be of the form
\begin{equation}
  \label{eq:rhod_form2}
  \rho^{(d)}(\e^{(d)}) = \biggl( 1 - \frac{d-1}{d}\e^{(d)}\biggr)\ket{1}\bra{1} + \frac{\e^{(d)}}{d} \sum_{i=2}^{d} \ket{i}\bra{i}.
\end{equation}
For \cref{eq:rhod_form1} and \cref{eq:rhod_form2} to describe the same state, they need to have the same coefficient for each diagonal entry.  
This requirement leads to three linear equations for $\e^{(d)}$ and $q$, 
in terms of $\e^{(2)}$.  But these three equations are linearly dependent, so it suffices to express them with the following two equations: 
\begin{align}
q = & \frac{d-2}{d} \e^{(d)}, &
\e^{(2)} &= \frac{2 \e^{(d)}}{d-(d-2) \e^{(d)}}, 
\end{align}
where the parameters $q$ and $\e^{(2)}$ are expressed in terms of $\e^{(d)}$ and $d$.  We note that for $\e^{(2)} \in (0,1)$, the corresponding $q$ and $\e^{(d)}$ are also in $(0,1)$.  

Let the final output fidelity be $F_{\mathrm{embed}}$ and $F_{\mathrm{opt}}$ for $\mathcal{P}_{\mathrm{embed}}^{(2)}$ and $\mathcal{P}_{\mathrm{opt}}^{(2)}$, respectively. Since $\mathcal{P}_{\mathrm{embed}}^{(2)}$ cannot outperform $\mathcal{P}_{\mathrm{opt}}^{(2)}$, we know that
$F_{\mathrm{embed}} \leq F_{\mathrm{opt}}$.
From the qubit case~\cite{CEM99}, the final fidelity satisfies $F_{\mathrm{opt}} \leq 1 - \frac{\e^{(2)}}{2(1-\e^{(2)})^2}\frac{1}{N}$, so if $F_{\mathrm{embed}} = 1 - \epsilon$, then $F_{\mathrm{opt}} \geq 1 -\epsilon$ and we have
\begin{equation}
  N \geq \frac{\e^{(d)}(d-(d-2)\e^{(d)})}{d^2(1-\e^{(d)})^2}\frac{1}{\epsilon}
  = \frac{\e^{(d)}}{(1-\e^{(d)})d\epsilon} + \Theta\biggl(\frac{1}{d^2\epsilon}\biggr).
\end{equation}
This shows that $ N \in \Omega(\frac{1}{d\epsilon})$.
\end{proof}

For constant $d$, the lower bound is $\Omega(\frac{1}{\epsilon})$, which matches our upper bound on the sample complexity of $\PURIFY$, showing that our approach is asymptotically optimal. 
As discussed further in the next section, the results of \cite{li2024optimal} indicate that the asymptotic scaling with $d$ in our lower bound is not optimal. 
Moreover, up to small multiplicative constants, the $\PURIFY$ procedure is asymptotically optimal in $d$ and $\epsilon$ for $\e < 1/3$.  For $\e \gg 2/3$, our analysis of the number of iterations $i^*$ to suppress the error down to $2/3$ suggests suboptimality of the dependence on $1-\e$ relative to the $\Omega((1-\e)^{-2})$ performance of the optimal protocol.

\section{Conclusions} \label{sec:con}

We have established a general purification procedure based on the swap test. This procedure is efficient, with sample complexity $O(\frac{1}{\epsilon})$ for any fixed initial error $\e$, independent of the dimension $d$. The simplicity of the procedure makes it easy to implement, and in particular, it requires a quantum memory of size only logarithmic in the number of input states consumed. Note however that it requires a quantum memory with long coherence time since partially purified states must be maintained when the procedure needs to restart.

It is natural to ask
whether the purification protocol can be asymptotically improved when the dimension is not fixed. One way of approaching this question, which is also of independent interest, is to characterize the optimal purification protocol. In a related line of work, partially reported in \cite{futhesis}, some of the authors developed a framework to study the optimal qudit purification procedure using Schur--Weyl duality. The optimal procedure can be characterized in terms of a sequence of linear programs involving the representation theory of the unitary group. Using the Clebsch--Gordan transform and Gel'fand--Tsetlin patterns, we gain insight into the common structure of all valid purification operations and give a concrete formula for the linear program behind the optimization problem. For qubits, this program can be solved to reproduce the optimal output fidelity
\begin{align}
    F_2 = 1 - \frac{1}{2(N+1)} \cdot \frac{\e}{(1-\e)^2} + O\Bigl(\frac{1}{N^2}\Bigr)
\end{align}
as established in Ref.~\cite{CEM99}. For qutrits, we show that the optimal output fidelity is
\begin{align}
    F_3 = 1 - \frac{2}{3(N+1)} \cdot \frac{\e}{(1-\e)^2} + O\Bigl(\frac{1}{N^2}\Bigr).
\end{align}
More generally, our formulation makes it easy to numerically compute the optimal fidelity and determine the optimal purification procedure for any given dimension. Due to the complexity of the Gel'fand--Tsetlin patterns, a general analytical expression for the optimal solution of the linear program was elusive,
but this problem was recently solved \cite{li2024optimal}.
It turns out that the optimal qudit fidelity is 
\begin{align}
    F_d = 1 - \frac{d-1}{d(N+1)} \cdot \frac{\e}{(1-\e)^2} + O\Bigl(\frac{1}{N^2}\Bigr),
\end{align}
so for $F_d \geq 1 - \epsilon$, the sample complexity is 
\begin{align}
    N \leq \frac{d-1}{d} \cdot \frac{\e}{\epsilon(1-\e)^2}. 
\end{align}
Comparing to the sample complexity in \cref{eq:SCn}, we see that our expected sample complexity is optimal in $d$ and $\epsilon$ up to a small multiplicative constant if $\e < 1/3$.

Simon's problem with a faulty oracle, as described in \cref{sec:simon}, is just one application of $\PURIFY$. This approach can also be applied to quantum algorithms involving parallel or sequential queries of depth $O(1)$ subject to a depolarizing oracle. However, for a quantum algorithm that makes  $\omega(1)$ sequential queries, the query overhead of our purification protocol is significant. In future work, it might be interesting to explore query complexity with faulty oracles more generally, or to develop other applications of $\PURIFY$ to quantum information processing tasks.

The swap test, which is a key ingredient in our algorithm, is a special case of generalized phase estimation~\cite[Section~8.1]{Aram}.
Specifically, the swap test corresponds to generalized phase estimation where the underlying group is $\S{2}$.
This places our algorithm in a wider context and suggests possible generalizations using larger symmetric groups $\S{k}$ for $k \geq 2$.
Such generalizations might provide a tradeoff between sample complexity and memory usage, since generalized phase estimation for $\S{k}$ operates simultaneously on $k$ samples.
It would be interesting to understand the nature of such tradeoffs for purification or for other tasks.
In particular, can one approximately implement the Schur transform using less quantum memory?

Our procedure has a pure target state $\proj{\psi}$ which is also its fixed point.
The same procedure cannot be used for mixed target states since its output might be ``too pure.''
Nevertheless, one can still ask how to implement 
the transformation 
\begin{equation}
\left( (1-\e) \rho + \e \frac{I}{d} \right)^{\otimes N} \mapsto \;
   (1-\epsilon) \rho + \epsilon \frac{I}{d}
\end{equation}
for $\epsilon \ll \e$
and for sufficiently large $N$.  This purification can be performed
using quantum state tomography.  For pure states in various regimes, our method is much more efficient than this. Is there a more sample-efficient method for mixed states? Generalization of the 
swap test to techniques such as generalized phase estimation may 
provide a solution, but we leave this question for future research.

Finally, a natural generalization of the current work is the
purification of a general unknown mixed state $\rho$ with the 
promise that the principal eigenvalue is non-degenerate.
Applications have been found in \cite{GPS24}, and ongoing
research is addressing this open problem. 

\section*{Acknowledgements}
We thank Daniel Grier, David Gross, Barbara Kraus, Meenu Kumari, Aadil Oufkir, Hakop Pashayan, Luke Schaeffer, and Nengkun Yu for very helpful discussions.

AMC, HF, DL, and VV received support from NSERC during part of this project.
AMC also received support from NSF (grants 1526380, 1813814, and QLCI grant OMA-2120757).
HF acknowledges the support of the NSF QLCI program (grant OMA-2016245).
DL and ZL, via the Perimeter Institute, are supported in part by the Government of Canada through the
Department of Innovation, Science and Economic Development and by the
Province of Ontario through the Ministry of Colleges and Universities.
MO was supported by an NWO Vidi grant (Project No.~VI.Vidi.192.109).

\bibliographystyle{alphaurl}
\bibliography{purification}

\clearpage
\appendix

\section{Output error parameter from swap test as a function of input errors} \label{sec:io}

Here, we analyse the output error parameter when the input states for the swap test gadget (\cref{alg:swap})
are not necessarily the same.  
For the purpose of purification, we are particularly interested in the special case when both input states, $\rho$ and $\sigma$, are of the form
\begin{equation}
  \rho(\e) := (1 - \e) \proj{\psi} + \e \, \frac{I}{d}
\end{equation}
for some \emph{error parameter} $\e \in (0,1)$, where $\ket{\psi} \in \C^d$ is an unknown but fixed $d$-dimensional state (the dimension $d \geq 2$ is often implicit).  
The following claim shows that the $\SWAP$ gadget preserves the form of these states, irrespective of the actual state $\ket{\psi}$, and derives the error parameter of the output state as a function of both input error parameters.

\begin{proposition}\label{claim:es}
If the input states of \cref{alg:swap} are $\rho(\e_1)$ and $\rho(\e_2)$ 
as described in \cref{eq:re},
then
\begin{equation}
  \SWAP \of[\big]{\rho(\e_1),\rho(\e_2)} = \rho(\e')
  \qquad
  \text{where}
  \qquad
  \e' := \frac{\frac{\e_1 + \e_2}{2} + \frac{\e_1 \e_2}{d}}
                    {(1 + \frac{1}{d}) + (1 - \frac{1}{d}) (1 - \e_1) (1 - \e_2)}.
  \label{eq:e'}
\end{equation}
\end{proposition}

\begin{proof}
\Cref{eq:re} shows that both input states are linear combinations of projectors $\set{\proj{\psi}, I}$.
Since this set is closed under matrix multiplication, the output state in \cref{eq:SWAP} is also a linear combination of $\proj{\psi}$ and $I$.
By construction, the output is positive semidefinite and has unit trace, so it must be of the form $\rho(\e')$ for some $\e' \in [0,1]$.

Note that $\rho(\e_1)$ and $\rho(\e_2)$ are diagonal in the same basis, so
$\rho(\e_1) \rho(\e_2) = \rho(\e_2) \rho(\e_1)$
and hence the numerator of \cref{eq:SWAP} is
$\rho(\e_1) + \rho(\e_2) + 2 \rho(\e_1) \rho(\e_2)$.
Using \cref{eq:re}, the coefficient of $I$ in this expression is
\begin{equation}
  \alpha := \frac{\e_1}{d} + \frac{\e_2}{d} + 2 \frac{\e_1 \e_2}{d^2}.
\end{equation}
The denominator of \cref{eq:SWAP} is given by
\begin{equation}
\begin{aligned}
  \beta
 & := 1 + \Tr \of[\big]{\rho(\e_1) \rho(\e_2)} \\
 &\:= 1 + (1 - \e_1) (1 - \e_2)
        + (1 - \e_1) \frac{\e_2}{d}
        + \frac{\e_1}{d} (1 - \e_2)
        + \frac{\e_1}{d} \frac{\e_2}{d} d \\
 &\:= 1 + \of[\Big]{1 - \frac{1}{d}} (1 - \e_1) (1 - \e_2)
        + \frac{1}{d}
          \of[\big]{(1 - \e_1) + \e_1}
          \of[\big]{(1 - \e_2) + \e_2} \\
 &\:= \of[\Big]{1 + \frac{1}{d}}
    + \of[\Big]{1 - \frac{1}{d}} (1 - \e_1) (1 - \e_2). \label{eq:beta}
\end{aligned}
\end{equation}
The coefficient of $I/d$ for the output state in \cref{eq:SWAP} is thus $\frac{d}{2} \cdot \frac{\alpha}{\beta}$, which agrees with $\e'$ in \cref{eq:e'}.
\end{proof}

\Cref{claim:es} tells us how the $\SWAP$ gadget affects the error parameter $\delta$.
In particular, the output state is purer than both input states (i.e., $\delta < \min \set{\delta_1, \delta_2}$) only when $\delta_1$ and $\delta_2$ are sufficiently close.
The region of values $(\delta_1,\delta_2)$ for which this holds is illustrated in \cref{fig:Region} for several different $d$.

\begin{figure}[H]
\centering
\includegraphics[width = 0.5\textwidth]{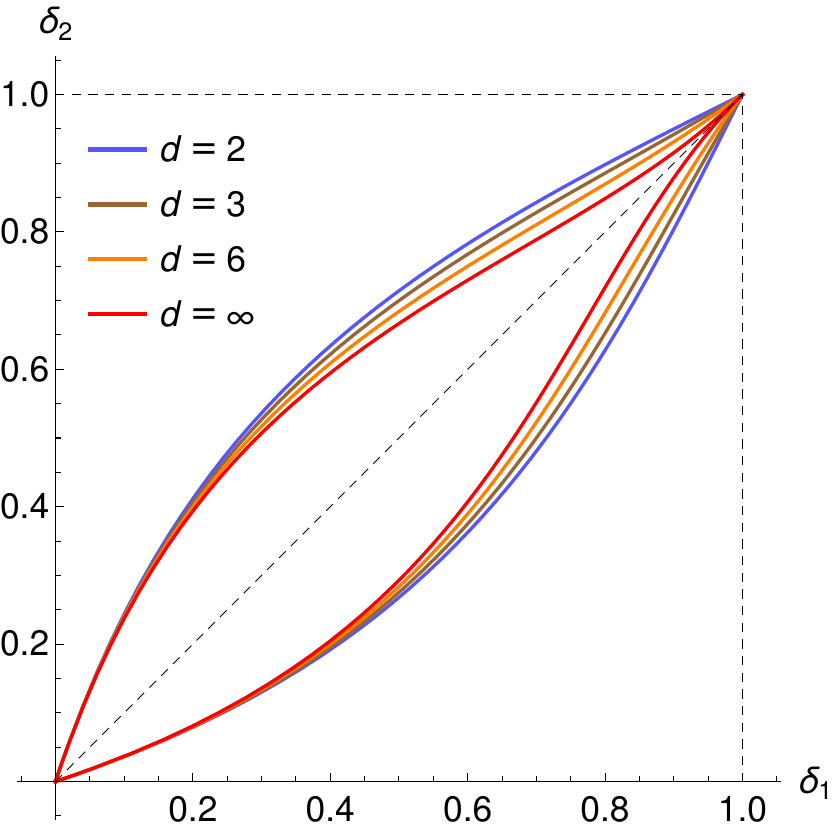}
\caption{\label{fig:Region}The outer boundary of the region in the $(\e_1,\e_2)$ plane where the $\SWAP$ gadget (\cref{alg:swap}) improves the purity of the output compared to both inputs (see \cref{claim:edif}), for dimensions $d \in \set{2,3,6,\infty}$. As $d$ increases, the region shrinks and converges to the area enclosed by the red curve as $d \to \infty$. For any $d$, the region contains a substantial area surrounding the $\e_1 = \e_2$ line, except for the tip at $\delta_1 = \delta_2 = 1$ which becomes infinitely sharp at $d = \infty$.}
\end{figure}

\begin{proposition}\label{claim:edif}
Let $\e_1, \e_2 \in (0,1)$ and assume (without loss of generality) that $\e_2 \geq \e_1$. The output of the $\SWAP$ gadget (see \cref{alg:swap}) is purer than both inputs $\rho(\e_1)$ and $\rho(\e_2)$ if and only if
\begin{equation}
  \e_2 - \e_1
  < (1 - \e_1)
    \frac{2 \e_1 (d - (d-1) \e_1)}
     {d + 2 \e_1 (d - (d-1) \e_1)}.
  \label{eq:e2e1}
\end{equation}
\end{proposition}

\begin{proof}
Since we assumed that $\e_2 \geq \e_1$, the output is purer than both inputs whenever $\e_1 > \e'$ where $\e'$ is given in \cref{claim:es}. The denominator of $\e'$ is always positive, so $\e_1 > \e'$ is equivalent to
\begin{equation}
  \e_1
  \of[\bigg]{
    \of[\Big]{1 + \frac{1}{d}} +
    \of[\Big]{1 - \frac{1}{d}} (1 - \e_1) (1 - \e_2)
  }
  > \frac{\e_1 + \e_2}{2} + \frac{\e_1 \e_2}{d}.
\end{equation}
If we multiply both sides by $2 d$ and expand some products, we get
\begin{equation}
  2 \e_1 (d + 1) + 2 \e_1 (d - 1) (1 - \e_1 - \e_2 + \e_1 \e_2)
  > d (\e_1 + \e_2) + 2 \e_1 \e_2,
\end{equation}
which can be simplified to an affine function in $\e_2$:
\begin{equation}
        \of[\big]{3 d \e_1 - 2 (d-1) \e_1^2}
  - \of[\big]{d + 2 d \e_1 - 2 (d-1) \e_1^2} \e_2 > 0.
\end{equation}
This inequality is equivalent to
\begin{equation}
  \frac{3 d \e_1 - 2 (d-1) \e_1^2}
   {d + 2 d \e_1 - 2 (d-1) \e_1^2} > \e_2,
   \label{eq:e2frac}
\end{equation}
where the denominator agrees with \cref{eq:e2e1} and is always positive since $\e_1 > \e_1^2$. If we subtract $\e_1$ from both sides of \cref{eq:e2frac}, the numerator becomes
\begin{equation}
\begin{aligned}
  \of[\big]{ 3 d \e_1 - 2 (d-1) \e_1^2} - \of[\big]{d + 2 d \e_1 - 2 (d-1) \e_1^2 } \e_1
  &= 2 d \e_1 - 2 (d-1) \e_1^2 - 2 d \e_1^2 + 2 (d-1) \e_1^3 \\
  &= 2 \e_1 (1 - \e_1) \of[\big]{ d - (d-1) \e_1 },
\end{aligned}
\end{equation}
which agrees with \cref{eq:e2e1}.
\end{proof}
Note that the right-hand side of \cref{eq:e2e1} is strictly positive for any $\e_1 \in (0,1)$. In particular, if $\e_1 = \e_2 \in (0,1)$ then the left-hand side vanishes while the right-hand side is strictly positive. This means that for identical inputs, the output is always\footnote{We exclude the extreme cases when both inputs are completely pure or maximally mixed, since in these cases the purity cannot be improved in principle: a pure state cannot get any purer, and the maximally mixed state has no information about the desired (pure) target state $\ket{\psi}$.} purer than both inputs, giving another proof of \cref{cor:better} (see also \cref{fig:Region}).

\end{document}